\documentclass[a4paper,superscriptaddress,10pt,shorttitle=profunctoroptics,accepted=2023-03-06,issue=1,volume=6,onecolumn]{compositionalityarticle}
\pdfoutput=1

\pdfoutput=1

\usepackage[utf8]{inputenc}
\usepackage[english]{babel}
\usepackage[T1]{fontenc}
\usepackage[svgnames,rgb]{xcolor}
\usepackage{amsfonts}
\usepackage{amssymb}
\usepackage{amsmath}
\usepackage{amsthm}
\usepackage{mathtools, cuted}
\usepackage[margin=1cm]{caption}
\allowdisplaybreaks{}
\usepackage{tikz}
\usetikzlibrary{matrix,arrows,shapes}
\usepackage{tikz-cd}
\usepackage{enumitem}
\setlist[itemize]{leftmargin=*}

\usepackage[pagebackref,breaklinks]{hyperref}
\definecolor{darklink}{RGB}{76,70,255}%
\hypersetup{
  pdftitle={Profunctor optics, a categorical update},
  pdfauthor={Clarke, Elkins, Gibbons, Loregian, Milewski, Pillmore, Román},
  pdfkeywords={optics, lens, tambara module},
  colorlinks = true,
  linkcolor = darklink,
}
\usepackage[nameinlink]{cleveref}

\usepackage[anythingbreaks]{breakurl}
    \expandafter\def\expandafter\UrlBreaks\expandafter{\UrlBreaks\do\a%
    \do\b\do\c\do\d\do\e\do\f\do\g\do\h\do\i\do\j\do\k\do\l\do\m\do\n%
    \do\o\do\p\do\q\do\r\do\s\do\t\do\u\do\v\do\w\do\x\do\y\do\z\do\&}

\long\def\ignore#1{}
\usepackage{adjustbox}
\usepackage{supertabular}
\usepackage{longtable}

\theoremstyle{plain}
\newtheorem{theorem}{Theorem}
\newtheorem{proposition}[theorem]{Proposition}

\newtheorem{lemma}[theorem]{Lemma}

\theoremstyle{definition}
\newtheorem{definition}[theorem]{Definition}

\theoremstyle{remark}
\newtheorem{remark}[theorem]{Remark}
\newtheorem{example}[theorem]{Example}

\begingroup\makeatletter
  \@for\theoremstyle:=definition,remark,plain
  \do{\expandafter\g@addto@macro\csname th@\theoremstyle\endcsname{\addtolength\thm@preskip\parskip}}
\endgroup

\numberwithin{theorem}{section}

\usepackage{array}
\newcolumntype{H}{>{\setbox0=\hbox\bgroup}c<{\egroup}@{}}
 \usepackage[numbers,sort&compress]{natbib}
\newcommand{\Optic}{\hyperlink{linkoptic}{\mathbf{Optic}}}
\newcommand{\Sets}{\mathbf{Set}}
\newcommand{\Lan}{\mathsf{Lan}}

\newcommand{\Nat}{\mathrm{Nat}}  
\newcommand{\Tamb}{\hyperlink{linktamb}{\mathbf{Tamb}}}
\newcommand{\Prof}{\mathbf{Prof}}
\newcommand{\C}{\mathbf{C}}
\newcommand{\D}{\mathbf{D}}
\newcommand{\E}{\mathbf{E}}
\newcommand{\M}{\mathbf{M}}
\newcommand{\N}{\mathbf{N}}
\newcommand{\X}{\mathbf{X}}
\newcommand{\Y}{\mathbf{Y}}
\newcommand{\A}{\mathbf{A}}   
\renewcommand{\AA}{\mathcal{A}}
\newcommand{\DA}{\mathcal{DA}}
\newcommand{\CCom}{\mathbf{CCom}}
\newcommand{\CMon}{\mathbf{CMon}}
\newcommand{\ListMonad}{\hyperlink{linklist}{\mathcal{L}}}

\newcommand*\circled[3]
  {\tikz[baseline={([yshift=#2]current bounding box.center)}]{
      \node[shape=circle,draw,inner sep=#1] (char) {#3};}}

\newcommand{\docircL}[3]{\mathbin{\circled{#1}{#2}{\scalebox{#3}{\tiny{$\mathsf{L}$}}}}}
\newcommand{\docircR}[3]{\mathbin{\circled{#1}{#2}{\scalebox{#3}{\tiny{$\mathsf{R}$}}}}}

\newcommand*{\circL}{
  \mathchoice{\docircL{0.75pt}{-0.60ex}{1}}
             {\docircL{0.75pt}{-0.60ex}{1}}
             {\docircL{0.3pt}{-0.45ex}{0.70}}
             {\docircL{0.1pt}{-0.35ex}{0.50}}}
\newcommand*{\circR}{
  \mathchoice{\docircR{0.75pt}{-0.60ex}{1}}
             {\docircR{0.75pt}{-0.60ex}{1}}
             {\docircR{0.3pt}{-0.45ex}{0.70}}
             {\docircR{0.1pt}{-0.35ex}{0.50}}}

\newcommand{\actL}[2]{#1 \circL{} #2}
\newcommand{\actR}[2]{#1 \circR{} #2}

\newcommand{\nto}{\nrightarrow}
\newcommand\id{\mathrm{id}}

\newcommand\EM{\operatorname{EM}}
\newcommand\Kl{\operatorname{Kl}}

\newcommand\Vt[1]{\(\V\)\mbox{-}\nobreak\hspace{0pt}#1}
\newcommand\Wt[1]{\(\W\)\mbox{-}\nobreak\hspace{0pt}#1}
\newcommand\V{{\hyperlink{linkvenrichment}{\mathcal{V}}}}
\newcommand\W{{\hyperlink{linkwenrichment}{\mathcal{W}}}}
\newcommand\U{{\mathcal{U}}}
\newcommand{\Obj}{\operatorname{Obj}}
\newcommand\hask[1]{\texttt{\detokenize{#1}}}
\makeatletter
 \newcommand{\nicelinktarget}[1]{\Hy@raisedlink{\hypertarget{#1}{}}}
\makeatother
\newcommand\defining[2]{\nicelinktarget{#1}{{\emph{#2}}}}
\usepackage{fancyvrb} 
\usepackage{float}
\usepackage{listings}

\lstdefinestyle{mystyle}{
    basicstyle=\ttfamily,
    columns=fixed,
    breakatwhitespace=false,         
    breaklines=true,                 
    captionpos=b,                    
    keepspaces=true,                 
    numbersep=5pt,                  
    showspaces=false,
    showstringspaces=false,
}
\newcommand{\authorsforheader}{Clarke, Elkins, Gibbons, Loregian, Milewski, Pillmore, Román}
\newcommand{\paperdoi}{https://doi.org/10.32408/compositionality-6-1}
\newcommand{\receiveddate}{2022-03-16}

\newcommand\hint[1]{\quad\mbox{\{#1\}}}
\newcommand\Yoneda{\hyperlink{linkyoneda}{Yoneda}}
\newcommand\Coyoneda{\hyperlink{linkcoyoneda}{Coyoneda}}
\newcommand\hintCoyoneda{\hint{\Coyoneda}}
\newcommand\hintYoneda{\hint{\Yoneda}}
\newcommand\hintContinuity{\hint{\hyperlink{linkcontinuity}{Continuity}}}
\newcommand\hintAdjunction[2]{\hint{\hyperlink{linkadjunction}{Adjunction} $#1 \dashv #2$}}
\newcommand\hintNaturalTransformation{\hint{\hyperlink{linknaturaltransformationsends}{Natural transformation}}}

\newcommand\Lens{\hyperlink{linklens}{\mathbf{Lens}}}
\newcommand\MndLens{\hyperlink{linkmonadiclens}{\mathbf{MndLens}}}
\newcommand\AlgLens{\hyperlink{linkalgebraiclens}{\mathbf{AlgLens}}}
\newcommand\LinearLens{\hyperlink{linklinearlens}{\mathbf{LinearLens}}}
\newcommand\Prism{\hyperlink{linkprism}{\mathbf{Prism}}}
\newcommand\Affine{\hyperlink{linkaffine}{\mathbf{Affine}}}
\newcommand\Setter{\hyperlink{linksetter}{\mathbf{Setter}}}
\newcommand\Grate{\hyperlink{linkgrate}{\mathbf{Grate}}}
\newcommand\mLens{\hyperlink{linkmlens}{\mathbf{MonLens}}}

\newcommand\Traversal{\hyperlink{linktraversal}{\mathbf{Traversal}}}
\newcommand\Pw{\hyperlink{linkpw}{\operatorname{Pw}}}
\newcommand\App{\hyperlink{linkapp}{\mathbf{App}}}

\newcommand\tLenses{\hyperlink{linkdefinelens}{Lenses}}

\author[Clarke]{Bryce Clarke}
\orcid{0000-0003-0579-2804}
\author[Elkins]{Derek Elkins}
\author[Gibbons]{Jeremy Gibbons}
\orcid{0000-0002-8426-9917}
\author[Loregian]{Fosco Loregian}
\orcid{0000-0003-3052-465X}
\author[Milewski]{Bartosz Milewski}
\orcid{0000-0001-9249-0104}
\author[Pillmore]{Emily Pillmore}
\author[Román]{Mario Román}
\orcid{0000-0003-3158-1226}

\title{Profunctor Optics: a Categorical Update}

\begin{document}

\maketitle
\date{}

\lstset{language=Haskell}

\begin{abstract}
  Optics are bidirectional data accessors that capture data
  transformation patterns such as accessing subfields or iterating over
  containers. Profunctor optics are a particular choice of representation supporting modularity, meaning that we can construct accessors for
  complex structures by combining simpler ones.
\iffalse
  Profunctor optics have been
  studied only using \(\Sets\) as the enriching category and in the
  \emph{non-mixed} case. However, functional programming languages are arguably
  better described by enriched categories and we have found that some structures
  in the literature are actually \emph{mixed} optics.
\else
  Profunctor optics have previously been
  studied only in an unenriched and non-mixed setting,
  in which both directions of access are modelled in \(\Sets\).
  However, functional programming languages are arguably
  better described by enriched categories;
  and we have found that some structures
  in the literature are actually \emph{mixed} optics,
  with access directions modelled in different categories.
\fi
  Our work generalizes a
  classic result by Pastro and Street on Tambara theory and uses it to describe
  \emph{mixed} \Vt{enriched} profunctor optics and to endow them with
  \Vt{category} structure. We provide some original families of optics and
  derivations, including an elementary one for \emph{traversals}. Finally, we discuss a Haskell implementation.
  \(\quad\)

\end{abstract}

\newpage
\tableofcontents
\newpage

\section{Introduction}

\subsection{Optics}

\emph{Optics} are an abstract representation of some common patterns in
bidirectional data accessing.
The most widely known optics are \emph{lenses},
`focusing' on a subfield $A$ of a larger data structure $S$ through
a pair of functions \emph{view} $(S \to A)$ and \emph{update}
$(S \times A \to S)$ that respectively retrieve and modify the field.
\emph{Lenses} have been used in functional programming as a compositional
solution to the problem of accessing fields of nested data
structures~\cite{foster05} (\Cref{fig:exampleaddress}).

\begin{figure}[H]
  \begin{lstlisting}[language=Haskell]
       data Address = Address
         { street'  :: String
         , town'    :: String
         , country' :: String }

       viewStreet :: Address -> String
       viewStreet = street'

       updateStreet :: Address -> String -> Address
       updateStreet a s = a {street' = s}

       example :: Address
       example = Address
         { street'  = "221b Baker Street"
         , town'    = "London"
         , country' = "UK" }

       >>> example
       Address { street'  = "221b Baker Street"
               , town'    = "London"
               , country' = "UK"}

       >>> viewStreet example
       "221b Baker Street"

       >>> updateStreet example "4 Marylebone Road"
       Address { street'  = "4 Marylebone Road"
               , town'    = "London"
               , country' = "UK"}
  \end{lstlisting}%
\caption{Lenses are pairs of `view' and `update' functions that capture
  the repeating pattern of accessing subfields.
  Here, \lstinline[language=Haskell]!viewStreet! extracts a field from a record, and
  \lstinline[language=Haskell]!updateStreet! updates that field.
}
\label{fig:example0}
\label{fig:exampleaddress}
\end{figure}

As the understanding of these data accessors grew, different \emph{families of optics}
were introduced for a variety of different types (e.g. \emph{prisms} for disjoint
unions and \emph{traversals} for containers), each one of them capturing a
particular data accessing pattern (\Cref{fig:examplePrism}).

\subsection{Modularity}
It is straightforward to compose two lenses, one given by $S \to A$ and
$S \times A \to S$ and the other given by $A \to X$ and $A \times X \to A$, 
in order to access nested subfields. However, explicitly writing down this
composition (or explaining it to a computer) can be tedious. Intercomposability
only becomes increasingly difficult as other data accessors enter the stage:
composing a \emph{lens} given by $S \to A$ and $S \times A \to S$ with a \emph{prism} given by
$A \to X + A$ and $X \to A$ can produce a function
$S \to X \times (X \to S) + S$, which is neither a \emph{lens} nor a \emph{prism} but a different
optic known as \emph{affine traversal}. Implementing explicitly a composition
like this for every possible pair of optics would be prone to errors and result
in a large codebase. However, we would like optics to behave \emph{modularly};
in the sense that, given two optics, it should be possible to join them into a
composite optic that directly accesses the innermost subfield.

\begin{figure}[H]
\centering
\begin{lstlisting}[language=Haskell] 
  buildString :: Address -> String
  buildString (Address s t c) = s ++ ", " ++ t ++ ", " ++ c

  verifyAddress :: String -> Either String Address
  verifyAddress a = case splitOn ", " a of
        [str, twn, ctr] -> Right (Address str twn ctr)
        failure -> Left a

  asAddress :: Prism Address String
  asAddress = mkPrism verifyAddress buildString

  >>> "221b Baker Street, London, UK" ?. asAddress
  Just (Address
    { street'  = "221b Baker Street"
    , town'    = "London" 
    , country' = "UK" })
\end{lstlisting}
\caption{A prism is given by a pair of functions \lstinline[language=Haskell]!match! and
\lstinline[language=Haskell]!build! that account for the possiblity of failure on pattern matching.
  In the figure, we verify whether a string can be parsed as an address.
  The combinator \lstinline[language=Haskell]!(?.)! returns the results using the \lstinline[language=Haskell]!Maybe! monad (see \S
  \ref{sec:tablecombinators}).}
\label{fig:examplePrism}
\end{figure}

Perhaps surprisingly, many implementations allow the programmer to wrap optics
into a different representation and then use \emph{ordinary function
  composition} to construct composite optics.

How is it possible to compose two constructs that are not functions using
ordinary function composition? Implementations provided by popular libraries
such as \textsf{lens} \cite{kmett15}, \textsf{mezzolens} \cite{oconnor15} in Haskell, or
\textsf{profunctor-optics} \cite{freeman15} in Purescript, achieve this effect by
using different representations of optics in terms of polymorphic functions and
the Yoneda lemma. This paper focuses on the encoding known as \emph{profunctor
  representation}, which is based on the isomorphism between lenses (and optics
in general) and functions that are polymorphic over profunctors with a particular
algebraic structure called a \emph{Tambara module}. Optics under this encoding
are called \emph{profunctor optics}.

\subsection{Profunctor optics}
Profunctor optics, and various other representations of lenses, were originally
proposed in functional programming as a compositional solution to the problem of
accessing fields of nested data structures~\cite{foster05,laarhoven09}.

Different families of profunctor optics are intercomposable. When we use the profunctor
representation of optics, composing optics of different kinds becomes also a
particular case of polymorphic function composition. Together, all \emph{families
  of optics} form a powerful language for modular data access. Consider the
example of \Cref{fig:lensandprism}, where a lens and a prism are used in
conjunction to manipulate parts of a string.

\begin{figure}[H]
\centering
  \begin{lstlisting}[language=Haskell]
  >>> let place = "221b Baker St, London, UK"

  >>> place ?. asAddress . street
  Just "221b Baker St"

  >>> place & asAddress . street .~ "4 Marylebone Rd"
    "4 Marylebone Rd, London, UK"
\end{lstlisting}
\caption{The composition of a prism (\texttt{asAddress}) and a lens
  (\texttt{street}) produces a composite optic (\texttt{asAddress} .
  \texttt{street}). This optic (an ``affine traversal'', see \S\ref{def:affine}) is
  used to parse a string and then access and modify one of its
  subfields.}\label{fig:lensandprism}
\end{figure}
 
Moreover, optics can be used to entirely change not only the value but the type
of the focus, and propagate that change back to the original data structure.
These are called \emph{type-variant} optics, in contrast with the
\emph{type-invariant} optics we have introduced so far (\Cref{fig:typevarying}).
In that case, the functions defining the optic need to account for that type
change (commonly, by also introducing polymorphism), but the internal
representation will work the same. The optics we discuss in this paper are
assumed to be type-variant, with type-invariant optics being a special case.

\begin{figure}[H]
\centering
\begin{lstlisting}
    data Timestamped a = Timestamped
      { created' :: UTCTime
      , modified' :: UTCTime
      , contents' :: a }

    viewContents :: Timestamped a -> a
    viewContents = contents'

    updateContents :: Timestamped a -> b -> Timestamped b
    updateContents x b = x {contents' = b}

    contents :: Lens' a b (Timestamped a) (Timestamped b)
    contents = mkLens' viewContents updateContents
\end{lstlisting}
\caption{A type-variant lens that targets the contents of a value paired with
  creation and modification timestamps. The lens is constructed from two functions \texttt{viewContents}
  and \texttt{updateContents}.}
\label{fig:typevarying}
\end{figure}

In its profunctor representation, each optic is written as a single function that
is polymorphic over profunctors with a certain algebraic structure. For
instance, \emph{lenses} can be written as functions polymorphic over
\emph{cartesian} profunctors, whereas \emph{prisms} can be written as functions
polymorphic over \emph{cocartesian} profunctors~\cite[\S 3]{pickering17}.
Milewski~\cite{milewski17} identified these algebraic structures (cartesian
profunctors, cocartesian profunctors,~\ldots) as \emph{Tambara
  modules}~\cite{tambara06} and used a result by Pastro and
Street~\cite{pastro08} to propose a unified definition of optic. This definition
has been later extended by Boisseau and Gibbons~\cite{boisseau18} and
Riley~\cite{riley18}, both using different techniques and proposing
laws for optics.

\subsection{Mixed profunctor optics}
However, the original result by Pastro and Street cannot be
used directly to unify all the optics that appear in practice. Our work
generalizes this result, going beyond the previous definitions of optic
to cover \emph{mixed}~\cite[\S 6.1]{riley18} and \emph{enriched optics}.

The generalized profunctor representation theorem captures optics already
present in the literature and makes it possible to upgrade them to more
sophisticated definitions. For instance, many generalizations of \emph{lenses}
in functional programming are shown to be particular cases of a more refined
definition that uses mixed optics (\Cref{def:lens}). We also show
derivations for some new optics that were not present in the literature
(\Cref{def:glass,def:kaleidoscope,def:algebraiclens}).
Finally, Milewski~\cite{milewski17} posed the problem of fitting the three basic
optics (lenses, prisms and traversals) into an elementary pattern; lenses and
prisms had been covered in his work, but traversals were missing. We present a
new description of \emph{traversals} in terms of power series functors (\Cref{prop:traversal}) whose derivation
is more direct than the ones based on \emph{traversables} as studied by
Jaskelioff and Rypacek~\cite{rypacek12}.

\subsection{Coend Calculus}

\emph{Coend calculus} is a branch of category theory that describes the behaviour of
\emph{ends} and \emph{coends}, certain universal objects associated with
profunctors \(P : \C^{op} \times \C \to \V\).  Ends can be thought of as some form of
universal quantifier, whereas coends can be thought of as their existential counterparts.

Ends are subobjects of the product
$\prod_{X \in \C} P(X,X)$, whereas coends result from quotienting the coproduct
$\coprod_{X \in \C}P(X,X)$. Both take into account the fact that $P$ depends on
two ``terms'', covariantly on the second, and contravariantly on the first.

\begin{definition}[Ends and coends]
  The \emph{end} is the equalizer of the action of morphisms on both
arguments of the profunctor, whereas the \emph{coend} is dually defined as a coequalizer.
\[\mathrm{end}(P) \coloneqq \mathrm{eq} \left( \begin{tikzcd}
\prod_{X \in \C} P(X,X) \rar[yshift=0.5ex]\rar[yshift=-0.5ex] & \prod_{f \colon A \to B} P(A,B)
\end{tikzcd}
\right),
\]
\[\mathrm{coend}(P) \coloneqq \mathrm{coeq}\left( \begin{tikzcd}
\coprod_{f \colon B \to A} P(A,B) \rar[yshift=-0.5ex, swap] \rar[yshift=0.5ex] &
\coprod_{X \in \C} P(X,X)
\end{tikzcd}\right).
\]
\end{definition}

\emph{Ends} are usually denoted with a subscripted integral;
\emph{coends} use a superscripted integral.
\[\int_{X \in \C} P(X,X) \coloneqq \mathrm{end}(P), \qquad
  \int^{X \in \C} P(X,X) \coloneqq \mathrm{coend}(P).\]
In both cases, $X$ is a dummy variable, and we consider
$\int_{X \in \C} P(X,X)$ and $\int_{Y \in \C} P(Y,Y)$ `equivalent modulo
$\alpha$-conversion'.
The notation draws on an analogy with elementary calculus. An integral $\int f(x)\;dx$
depends ``covariantly'' on the variable $x$ defined, say, on $\mathbb{R}^{n}$, whereas
the differential ``$dx$'' can be regarded as an element of the \emph{dual} space
$(\mathbb{R}^{n})^{\ast}$. An even more striking analogy is that co/ends satisfy
a form of `Fubini rule' and a `Dirac delta' integration rule (see
\Cref{prop:fubinirule} and \Cref{prop:yonedareduction}).

Theorems involving ends and coends can be proved using their universal
properties. Here, we offer a terse account of coend calculus
\cite{caccamo01,loregian19}. Using the calculus based on the following rules, it
is possible to construct isomorphisms between objects of a category by means of
a chain of `deduction steps'.

\begin{proposition}
  Evaluation on the identity defines the following isomorphisms, called
  \defining{linkyoneda}{Yoneda} and \defining{linkcoyoneda}{coYoneda} reductions, respectively.
\label{prop:yonedareduction}
  \[ \int_{X \in \C} \V(\C(A,X),FX) \cong  FA. \qquad
    \int^{X \in \C} \C(X,A) \otimes FX \cong FA, \]
where \(\otimes \) is the tensor product in the monoidal category \(\V \) 
 (the base of the enrichment for \(\C \)) and $F \colon \C \to \V$ 
is a co-presheaf (there are analogous identities for presheaves).
\end{proposition}

\begin{proposition}
  The
  \defining{linkfubini}{Fubini rule}
  is satisfied up to isomorphism.
\label{prop:fubinirule}
    \[\int_{X_{1} \in \C} \int_{X_{2} \in \C} P(X_{1},X_{2},X_{1},X_{2})
    \cong
    \int_{X_{2} \in \C} \int_{X_{1} \in \C}  P(X_{1},X_{2},X_{1},X_{2}).\]
  \[\int^{X_{1} \in \C} \int^{X_{2} \in \C} P(X_{1},X_{2},X_{1},X_{2})
    \cong
    \int^{X_{2} \in \C} \int^{X_{1} \in \C}  P(X_{1},X_{2},X_{1},X_{2}).\]
\end{proposition}

\begin{proposition}
  \label{prop:continuity}
  \defining{linkcontinuity}{Continuity} and \defining{linkcocontinuity}{cocontinuity} induce the following isomorphisms.
  \[\V\left(A, \int_{X \in \C} P(X,X)\right) \cong \int_{X \in \C} \V(A,P(X,X)).\]
  \[\V\left(\int^{X \in \C} P(X,X), A\right) \cong \int_{X \in \C} \V(P(X,X),A).\]
\end{proposition}

\begin{proposition}
  \label{prop:naturaltransformationsasends}
  The set of \defining{linknaturaltransformationsends}{natural transformations} can be rewritten as a coend.
  \[\int_{X \in \C}\D(FX,GX) \cong [\C,\D](F,G) \phantom{\int}\]
\end{proposition}

In particular, coend calculus expressions are commonly simplified using \defining{linkadjunction}{adjunctions}
$(F \dashv G)$: an adjunction is equivalently an isomorphism $\D(FX,Y) \cong \C(X,GY)$.

\subsection{Contributions}

Our first contribution is the derivation and partial classification of mixed
optics, covering both optics existing in the literature and some novel ones, all
following a unified definition (\Cref{def:optic}). Our work completes and extends
the classification of (non-mixed) optics in \cite{boisseau18,riley18}.

Explicitly, we present a new family of optics in \Cref{def:algebraiclens}, that
unifies new examples with some optics already present in the literature, such as
\emph{achromatic lenses} \cite[\S 5.2]{boisseau17}. We introduce an original
derivation showing that \emph{monadic lenses}~\cite{AbouSaleh16} are mixed
optics in \Cref{prop:monadiclens}. Similarly, in \Cref{prop:monoidallenses}, we
present a novel derivation showing that the appropiate generalization of lenses
to an arbitrary monoidal category \cite[\S 2.2]{spivak19} is not an optic but a
mixed optic. We give a unified definition of \emph{lens} in \Cref{def:lens},
that for the first time can be specialized to all of these previous examples.
Finally, we present a new derivation of the optic known as \emph{traversal} in
\Cref{prop:traversal}.

Our second contribution is the definition of the enriched category of mixed profunctor
optics. The construction requires a generalization of the \emph{Tambara modules}
of \cite{pastro08} that had been used to define categories of profunctor optics
\cite{boisseau18,riley18} to \emph{generalized Tambara module}.
This is done in \Cref{sec:tambaratheory}. As a corollary, we extend the result
that justifies the use of the profunctor representation of optics in functional
programming to the case of enriched and mixed optics (\Cref{th:profrep}),
endowing them with \Vt{category} structure.

\subsection{Synopsis}

We introduce the definition of \emph{mixed optic} in \Cref{sec:optics}.
\Cref{sec:examples} describes some practical examples from functional
programming and shows how they are captured by the definition.
\Cref{sec:tambaratheory} describes how the theory of Tambara modules can be
applied to obtain a profunctor representation for optics. \Cref{sec:conclusions}
contains concluding remarks. The Appendix (\Cref{sec:implementation}) introduces
the details of a full Haskell implementation.

\subsection{Setting}

We shall work with categories enriched over a B\'enabou cosmos \((\defining{linkvenrichment}{\ensuremath{\mathcal{V}}},\otimes,I)\);
that is, a (small)-complete and cocomplete symmetric monoidal closed category.
In particular, \(\V\) is enriched over itself, and we write the internal
hom-object between \(A,B \in \operatorname{Obj}(\V)\) as
\(\V(A,B)\). Our intention is to keep a close eye
on the applications in functional programming: the enriching category \(\V\)
should be thought of as the category whose objects model the types of an
idealized programming language and whose morphisms model the programs. Because
of this, \(\V\) will be cartesian in many of the examples. We can, however,
remain agnostic as to which specific \(\V\) we are addressing.

For calculations, we make significant use of coend calculus as described, for
instance, by Loregian~\cite{loregian19}. The proofs in this paper can be carried
out without assuming choice or excluded middle, but there is an important
set-theoretical issue: in some of the examples, we compute coends over non-small
categories. We implicitly fix a suitable Grothendieck universe and our
categories are to be considered small with respect to that universe. As
Riley~\cite[\S 2]{riley18} notes, this will not be a problem in general: even if
some coends are indexed by large categories and we cannot argue their existence
using the cocompleteness of \(\Sets\), we will still find them represented by
objects in the category.

\section*{Acknowledgements}
This work was started in the last author’s MSc thesis~\cite{roman19};
development continued at the Applied Category School 2019 at
Oxford~\cite{pillmore20}, and we thank the organizers of the School for that
opportunity. We also thank Pawe\l{} Soboci\'nski for discussion,
and the anonymous reviewers for suggestions
about previous versions of this manuscript. The code for
this text has been continued as a Haskell library \cite{vitrea20}.

Bryce Clarke is supported by the Australian Government Research Training Program
Scholarship. Fosco Loregian and Mario Román were supported by the European Union
through the ESF funded Estonian IT Academy research measure (project
2014-2020.4.05.19-0001).

\section{Optics}
\label{sec:optics}

Our first goal is to give a unified definition that captures what it means to be
an \textbf{optic}. We have seen so far how \emph{lenses} and \emph{prisms} work
(\Cref{fig:exampleaddress,fig:examplePrism}). A common pattern can be
extracted from these two cases. \emph{Lenses} can be constructed when we have a
function that splits some data structure of type $S$ into something of the form
$M \times A$ and then recombines it back to $S$. Here, $A$ is the type of
the field we want to focus on, and $M$ is the type combining the remaining
fields that constitute $S$. From this split, we can extract the pair of functions
$\C(S,A) \times \C(S \times A, S)$ that define a lens. \emph{Prisms} can be
constructed when we have a function that can split some data structure of type $S$
into something of the form $M + A$ and put the pieces together again.

The structure that is common to all optics is that they split a bigger data
structure of type \(S \in \C\) into some \emph{focus} of type \(A \in \C\) and
some \emph{context} or \emph{residual} \(M \in \M\) around it. We cannot access
the context directly, but we can still use its shape to update the original data
structure, replacing the current focus by a new one. The definition will capture
this fact imposing a quotient relation on the possible contexts; this quotient
is expressed by the dinaturality condition of a coend. The category of contexts
\(\M\) will be monoidal, allowing us to compose optics with contexts \(M\) and
\(N\) into an optic with context \(M \otimes N\). Finally, as we want to capture
\emph{type-variant} optics, we leave open the possibility of the new focus being of a
different type \(B \in \D\), possibly in a different category, which yields a
new data structure of type \(T \in \D\). This is summarized in \Cref{fig:opticstructure}.

\begin{figure}[H]
  \centering
\begin{tikzpicture}[x=0.75pt,y=0.75pt,yscale=-1,xscale=1]
\draw   (40,80) .. controls (40,74.48) and (44.48,70) .. (50,70) .. controls (55.52,70) and (60,74.48) .. (60,80) .. controls (60,85.52) and (55.52,90) .. (50,90) .. controls (44.48,90) and (40,85.52) .. (40,80) -- cycle ;
\draw    (60,80) -- (87,80) ;
\draw [shift={(90,80)}, rotate = 180] [fill={rgb, 255:red, 0; green, 0; blue, 0 }  ][line width=0.08]  [draw opacity=0] (10.72,-5.15) -- (0,0) -- (10.72,5.15) -- (7.12,0) -- cycle    ;
\draw   (90,80) .. controls (90,74.48) and (94.48,70) .. (100,70) .. controls (105.52,70) and (110,74.48) .. (110,80) .. controls (110,85.52) and (105.52,90) .. (100,90) .. controls (94.48,90) and (90,85.52) .. (90,80) -- cycle ; \draw   (92.93,72.93) -- (107.07,87.07) ; \draw   (107.07,72.93) -- (92.93,87.07) ;
\draw   (260,80) .. controls (260,74.48) and (264.48,70) .. (270,70) .. controls (275.52,70) and (280,74.48) .. (280,80) .. controls (280,85.52) and (275.52,90) .. (270,90) .. controls (264.48,90) and (260,85.52) .. (260,80) -- cycle ;
\draw   (210,80) .. controls (210,74.48) and (214.48,70) .. (220,70) .. controls (225.52,70) and (230,74.48) .. (230,80) .. controls (230,85.52) and (225.52,90) .. (220,90) .. controls (214.48,90) and (210,85.52) .. (210,80) -- cycle ; \draw   (212.93,72.93) -- (227.07,87.07) ; \draw   (227.07,72.93) -- (212.93,87.07) ;
\draw    (230,80) -- (257,80) ;
\draw [shift={(260,80)}, rotate = 180] [fill={rgb, 255:red, 0; green, 0; blue, 0 }  ][line width=0.08]  [draw opacity=0] (10.72,-5.15) -- (0,0) -- (10.72,5.15) -- (7.12,0) -- cycle    ;
\draw    (100,70) .. controls (100.42,42.15) and (213.07,40.77) .. (219.7,67.44) ;
\draw [shift={(220,70)}, rotate = 248] [fill={rgb, 255:red, 0; green, 0; blue, 0 }  ][line width=0.08]  [draw opacity=0] (10.72,-5.15) -- (0,0) -- (10.72,5.15) -- (7.12,0) -- cycle    ;
\draw   (130,110) .. controls (130,104.48) and (134.48,100) .. (140,100) .. controls (145.52,100) and (150,104.48) .. (150,110) .. controls (150,115.52) and (145.52,120) .. (140,120) .. controls (134.48,120) and (130,115.52) .. (130,110) -- cycle ;
\draw   (170,110) .. controls (170,104.48) and (174.48,100) .. (180,100) .. controls (185.52,100) and (190,104.48) .. (190,110) .. controls (190,115.52) and (185.52,120) .. (180,120) .. controls (174.48,120) and (170,115.52) .. (170,110) -- cycle ;
\draw    (100,90) .. controls (100.13,108.13) and (117.5,109.71) .. (127.12,109.95) ;
\draw [shift={(130,110)}, rotate = 180.97] [fill={rgb, 255:red, 0; green, 0; blue, 0 }  ][line width=0.08]  [draw opacity=0] (10.72,-5.15) -- (0,0) -- (10.72,5.15) -- (7.12,0) -- cycle    ;
\draw    (190,110) .. controls (208.33,110.11) and (218.05,102.14) .. (219.74,92.93) ;
\draw [shift={(220,90)}, rotate = 465] [fill={rgb, 255:red, 0; green, 0; blue, 0 }  ][line width=0.08]  [draw opacity=0] (10.72,-5.15) -- (0,0) -- (10.72,5.15) -- (7.12,0) -- cycle    ;
\draw (50,80) node    {$S$};
\draw (270,80) node    {$T$};
\draw (140,110) node    {$A$};
\draw (180,110) node    {$B$};
\draw (51,52.4) node [anchor=north west][inner sep=0.75pt]    {$\mathit{input}$};
\draw (133,23.4) node [anchor=north west][inner sep=0.75pt]    {$\mathit{residual}$};
\draw (71,112.4) node [anchor=north west][inner sep=0.75pt]    {$\mathit{focus}$};
\draw (211,113.4) node [anchor=north west][inner sep=0.75pt]    {$\mathit{new\ focus}$};
\draw (231,53.4) node [anchor=north west][inner sep=0.75pt]    {$\mathit{output}$};
\end{tikzpicture}
\caption{The common structure of an optic. We could provide semantics for diagram of this kind
  in terms of profunctors, see \cite{roman21}.}
\label{fig:opticstructure}
\end{figure}

Multiple definitions of optics of increasing generality have been given in
\cite{milewski17,boisseau18,riley18}. We encompass all of them under an abstract
definition in terms of monoidal actions.

Let \((\M,\otimes,I,a,\lambda,\rho)\) be a monoidal \Vt{category}
  \cite{day70}. Let it act on two arbitrary \Vt{categories} \(\C\) and \(\D\)
  by means of strong monoidal \Vt{functors} \((\actL{}{}) \colon \M \to [ \C , \C ]\)
  and \((\actR{}{}) \colon \M\to [\D ,\D]\); and let us write
  \[\begin{aligned}
      & \phi_{A}^L \colon A \cong \actL{I}{A}, &\quad& \phi_{M,N,A}^L \colon \actL{M}{\actL{N}{A}} \cong \actL{(M \otimes N)}{A}, \\
      & \phi_{B}^R \colon B \cong \actR{I}{B}, &\quad& \phi_{M, N,B}^R \colon \actR{M}{\actR{N}{B}} \cong \actR{(M \otimes N)}{B},
    \end{aligned}\]
  for the structure isomorphisms of the monoidal
  actions \(\actL{}{}\) and \(\actR{}{}\), which we use in infix notation.

\begin{definition}
  \label{def:optic}
Let \(S,A \in \C\) and \(T,B \in \D\). An $(\actL{}{},\actR{}{})$\textbf{-optic} from \((S,T)\) with the
focus on \((A,B)\) is a (generalized) element of the following object described as a coend:
\[\defining{linkoptic}{\ensuremath{\mathbf{Optic}}}_{\actL{}{},\actR{}{}} ((A, B), (S,T)) \coloneqq
\int^{M \in \mathbf{M}} \C(S, \actL{M}{A}) \otimes \D(\actR{M}{B}, T).
\]
\end{definition}

The two monoidal actions \(\actL{}{}\) and \(\actR{}{}\) represent the
two different ways in which the context interacts with the focus: one when the
data structure is decomposed and another one, possibly different, when it is
reconstructed. By varying these two actions we will recover many examples from the
literature and introduce some new ones, as the table in \Cref{table:optics} summarizes.

\begin{figure}[h]
\centering
\begin{tabular}{lllcH@{\hspace*{-\tabcolsep}}}
\hline
\vrule width0pt height2.5ex depth 1ex %
Name & Description & Actions & Base & Ref. \\
\hline
\vrule width0pt height2.5ex %
  Adapter
     &  \(\C(S,A) \otimes \D(B,T)\) & \((\Optic_{\id,\id})\) & $\V,\otimes$ & {\ref{def:adapter}} \\
  \hyperlink{linklens}{Lens}
     & \(\C(S,A) \times \D(S \bullet B, T)\)  & \((\Optic_{\times,\bullet})\) & $\W,\times$ &  {{\ref{def:lens}}} \\
  Monoidal lens
     & \(\CCom(S,A) \times \C(\U S \otimes B, T)\) & \((\Optic_{\otimes,_\U\times})\) & $\W,\times$ &  {{\ref{def:myerslens}}} \\
  Algebraic lens
     & \(\C(S,A) \times \D(\Psi S \bullet B, T)\) & \((\Optic_{_\U\times,_\U\bullet})\) &  $\W,\times$&  {{\ref{def:algebraiclens}}} \\
  Monadic lens
     & \(\W(S , A) \times \W(S \times B , \Psi T)\) & \((\Optic_{\times,\rtimes})\) & $\W,\times$& {{\ref{def:monadiclens}}} \\
  Linear lens
     & \(\C(S, [B, T] \bullet A)\)  & \((\Optic_{\bullet, \otimes})\) & $\V,\otimes$ & {{\ref{def:linearlens}}} \\
  \hyperlink{linkprism}{Prism}
     & \(\C(S, T \bullet A) \times \D(B,T)\) & \((\Optic_{\bullet,+})\) & $\W,\times$ & {{\ref{def:prism}}} \\
  Coalg. prism
     & \(\C(S, \Theta T \bullet A) \times \D(B,T)\) & \((\Optic_{_\U\bullet,_\U+})\) & $\W,\times$ & {{\ref{remark:coalgebraicprism}}} \\
  \hyperlink{linkgrate}{Grate}
     & \(\D([S,A] \bullet B, T)\) & \((\Optic_{\{\ ,\ \},\bullet})\) & $\V,\otimes$ & {{\ref{def:grate}}} \\
  \hyperlink{linkglass}{Glass}
     & \(\C(S \times [[S, A], B] , T)\) &  \((\Optic_{\times[\ ,\ ],\times[\ ,\ ]})\) & $\W,\times$ & {{\ref{def:glass}}} \\
  Affine traversal
     & \(\C(S, T + A \otimes \{B, T\})\) & \((\Optic_{+\otimes,+\otimes})\) & $\W,\times$ &  {{\ref{def:affine}}} \\
  \hyperlink{linktraversal}{Traversal}
     & \(\V(S, \sum\nolimits^n A^n \otimes [B^n, T])\) & \((\Optic_{\mathrm{Pw},\mathrm{Pw}})\) & $\V,\otimes$ & {{\ref{def:traversal}}} \\
  Kaleidoscope
     & \(\sum\nolimits_n \V([A^n,B],[S^n, T])\) & \((\Optic_{\mathrm{App},\mathrm{App}})\) & $\V,\otimes$ & {{\ref{def:kaleidoscope}}} \\
  Setter
     & \(\V([A,B],[S,T])\) & \((\Optic_{\mathrm{ev},\mathrm{ev}})\) & $\V,\otimes$ & {{\ref{def:setter}}} \\
  Fold
     & \(\V(S, \mathcal{L} A)\) & \((\Optic_{\mathrm{Foldable},\ast})\) & $\V,\otimes$ & {{\ref{def:fold}}}
\end{tabular}
\caption{Table of optics, together with their explicit description and their
  generating monoidal actions. The bullet represents some monoidal action; the brackets represent some monoidal action from the opposite category.}
\label{table:optics}
\end{figure}

The purpose of an abstract unified definition is twofold: firstly, it provides a
framework to classify existing optics and explore new ones, as we do in
\Cref{sec:examples}; and secondly, it enables a unified profunctor
representation, which we present in \Cref{sec:tambaratheory}.

\section{Examples of optics}
\label{sec:examples}

\subsection{Lenses and prisms}

\emph{Lenses} can be seen as accessors for a particular subfield of a data
structure. In its basic form, they are given by a pair of functions: the
\texttt{view} function that accesses the subfield; and
the \texttt{update} function that overwrites its
contents.

The basic definition of \emph{lens}~\cite{oles82,foster05,palmer07} has been generalized in many different directions. \emph{Monadic lenses}~\cite{AbouSaleh16},
\emph{lenses in a symmetric monoidal category}~\cite[\S 2.2]{spivak19},
\emph{linear lenses}~\cite[\S 4.8]{riley18} or \emph{achromatic lenses}~\cite[\S
5.2]{boisseau17} are some of them. These generalizations were not meant to be
mutually compatible. They use the monoidal structure in different
ways and introduce monadic effects in different parts of the signature. Some
were not presented as \emph{optics}, and a profunctor representation for them
was not considered. We present a unified description that captures all of these
variants.

We present two derivations of lenses as mixed optics that capture all of the
variants mentioned before, together with new ones, and endow them with a unified profunctor
representation (\Cref{th:profrep}).

The first derivation is based on cartesian structure. It generalizes the original
derivation~\cite{milewski17} and captures \emph{lenses in a symmetric monoidal
  category} (\Cref{def:myerslens}) and \emph{monadic lenses}
(\Cref{def:monadiclens}). It is then refined to cover \emph{achromatic
  lenses} and describe some new variants of optics missing in the literature. The
second derivation slightly generalizes \emph{linear lenses}~\cite[\S
4.8]{riley18} and uses the closed structure instead.

\subsubsection{Lenses}
Most variants of lenses rely on a cartesian monoidal structure. Throughout this
section, we take a cartesian closed category \((\defining{linkwenrichment}{\ensuremath{\mathcal{W}}}, \times,1)\) as our
base for enrichment. During the rest of the paper we will explicitly use $\W$ to
refer to a cartesian monoidal base of enrichment and use $\V$ to refer to a
not-necessarily-cartesian base of enrichment. A monoidal \Wt{category}
\((\C,\times,1)\) is said to be \emph{cartesian} if there exist \Wt{natural} isomorphisms
\(\C(Z,X \times Y) \cong \C(Z,X) \times \C(Z,Y)\) and \(\C(X,1) \cong 1\).

\begin{definition}\label{def:lens}
  Let \(\C\) be a cartesian \Wt{category} with a monoidal \Wt{action}
  \((\bullet) \colon \C \times \D \to \D\) to an arbitrary \Wt{category} \(\D\).
  A \defining{linkdefinelens}{lens} is an element of
  \[\defining{linklens}{\ensuremath{\mathbf{Lens}}}((A, B), (S, T)) \coloneqq
  \C(S,A) \times \D(S \bullet B, T).
  \]
\end{definition}
\begin{proposition}\label{prop:lens}
  \tLenses{} are mixed optics (as in \Cref{def:optic}) for the actions of
  the cartesian product \((\times) \colon \C \times \C \to \C\) and
  \((\bullet) \colon \C \times \D \to \D\). That is,
  \(\Lens \cong \Optic_{(\times,\bullet)}\).
\end{proposition}
\begin{proof}
  The universal property of the product can be summarized as it being right
  adjoint to the diagonal functor \((\Delta) \colon \C \to \C^{2}\).
  \begin{align*}
  & \int^{C \in \C} \C(S,C \times A) \times \D(C \bullet B, T) \\
  \cong & \hintAdjunction{\Delta}{(\times)} \\
  & \int^{C \in \C} \C(S , C) \times \C(S , A) \times \D(C \bullet B , T) \\
  \cong & \hintCoyoneda \\
  & \C(S , A) \times \D(S \bullet B , T). \qedhere
  \end{align*}
\end{proof}

This definition can be specialized to the pair
\(\C(S,A) \times \C(S \times B , T)\) if we take \(\C = \D\) and we let
\((\bullet)\) be the cartesian product. It is, however, more general than that, and it
captures the following examples (\Cref{def:myerslens,def:monadiclens}).

\subsubsection{Lenses in a symmetric monoidal category}

\begin{definition}\label{def:myerslens}
  A \textbf{monoidal lens} {{\cite[\S 2.2]{spivak19}}} in a symmetric monoidal category \(\C\) is a \emph{view and update} pair where
  the view is a commutative comonoid homomorphism,
  \[\defining{linkmlens}{\ensuremath{\mathbf{MonLens}_{\otimes}}} ((A, B), (S,T)) \coloneqq
    \CCom(S,A) \times \C(\U S \otimes B, T).\]
  Here \(\U\) represents the forgetful functor \(\U \colon \CCom \to \C\).
\end{definition}

\begin{proposition}\label{prop:monoidallenses}
  Monoidal lenses in a symmetric monoidal category are a particular case of \Cref{def:lens}.
\end{proposition}
\begin{proof}
  The category of cocommutative comonoids \(\CCom\) over a category \(\C\) can
  be given a cartesian structure in such a way that the forgetful functor
  \(\U \colon \CCom \to \C\) is strict monoidal (Fox shows a
  stronger result \cite{fox76}). By Proposition \ref{prop:lens}, we can show
  \(\mLens_\otimes \cong \Optic_{(\otimes,\bullet)}\) where
  \((\bullet)\) is given by \(S \bullet A \coloneqq \U S \otimes A\).
\end{proof}

\subsubsection{Monadic lenses}
\begin{definition}\label{def:monadiclens}
  Monadic lenses \cite[\S 2.3]{AbouSaleh16} allow for monadic effects in the
  \texttt{update} function (an example is Figure \ref{fig:exampleBox}). For
  \(\Psi \colon \W \to \W\) a \Wt{monad},
  \[\defining{linkmonadiclens}{\ensuremath{\mathbf{MndLens}_\Psi}} ((A, B), (S, T)) \coloneqq \W(S , A) \times \W(S \times B , \Psi T).\]
\end{definition}

\begin{proposition}\label{prop:monadiclens}
  Monadic lenses are a particular case of \Cref{def:lens}.
\end{proposition}
\begin{proof}
  Every \Wt{endofunctor} is \textit{strong}: this is because a functor $F \colon \W \to \W$ with a strength is the same thing as a $\W$-enrichment of the functor \cite{levy2019strong}. Thus, the
  \Wt{monad} \(\Psi\) comes with a \Wt{natural} family
  \(\theta_{X,Y} \colon X \times \Psi(Y) \to \Psi(X \times Y)\). This induces a
  \Wt{action}
  \((\rtimes) \colon \W \times \operatorname{Kl}(\Psi) \to \operatorname{Kl}(\Psi)\),
  with $\operatorname{Kl}(\Psi)$ the Kleisli category of the monad; defined, on morphisms, as the composite
\[\begin{tikzcd}[row sep=tiny]
   \W(X,Y) \times   \W(A,\Psi B)\rar{(\times)} &
    \W(X \times A, Y \times \Psi B) \\
    \phantom{\W(X \times A, Y \times \Psi B)} \rar{\theta} &
    \W(X \times A, \Psi (Y \times B)).
  \end{tikzcd}\]
  Using that
\(\operatorname{Kl}_\Psi(S \rtimes B,T) \coloneqq \W(S \times B, \Psi T)\),
monadic lenses are lenses (as in \Cref{def:lens}), \(\MndLens_\Psi \cong \Optic_{(\times,\rtimes)}\),
where \((\bullet)\) is given by
\((\rtimes) \colon \W \times \operatorname{Kl}(\Psi) \to \operatorname{Kl}(\Psi)\).
\end{proof}

\begin{remark}
  This technique is similar to the one used by \cite[\S 4.9]{riley18} to describe a non-mixed
  variant called \emph{effectful lenses}.
\end{remark}

\begin{figure}[h]
\centering
  \begin{lstlisting}
stamp :: MonadicLens IO a b (Timestamped a) (Timestamped b)
stamp = mkMonadicLens @IO viewContents updateContentsAndStamp
   where
     viewContents :: Timestamped a -> a
     viewContents = contents'

     updateStamp :: Timestamped a -> b -> IO (Timestamped b)
     updateStamp x b = do
       currentTime <- getCurrentTime
       return (x {contents' = b , modified' = currentTime})

greeting :: IO (Timestamped String)
greeting = do
  t <- getCurrentTime
  x <- pure (Timestamped t t "What is the answer?")
  threadDelay (2500000) -- microseconds
  x & stamp .! "42"

 >> greeting
 Contents: "42",
 Created:  2020-02-02 12:24:55.119075225 UTC,
 Modified: 2020-02-02 12:24:57.621826372 UTC.
\end{lstlisting}
\caption{A polymorphic family of type-variant monadic lenses for the \hask{IO}
  monad is used to track how a data holder (\hask{Timestamped}) is accessed 2.5
  seconds after creation.}
\label{fig:exampleBox}
\end{figure}

\subsubsection{Algebraic lenses}

We can further generalize \Cref{def:lens} if we allow the
\textit{context} over which we take the coend to be an algebra for a fixed
monad. The motivation is that lenses with a context like this appear to have
direct applications in programming; for instance, the \textit{achromatic
  lenses} of \cite{boisseau18} are a particular case of this definition. These
\textit{algebraic lenses} should not be confused with the previous
\textit{monadic lenses} in \Cref{def:monadiclens}.

\begin{definition}\label{def:algebraiclens}
  Let \(\Psi \colon \C \to \C\) be a \Wt{monad} on a cartesian \Wt{category}
  \(\C\). Let \((\bullet) \colon \C \times \D \to \D\) be a monoidal \Wt{action}
  on an arbitrary \Wt{category} \(\D\). An \textbf{algebraic lens} is an element
  of
\[\defining{linkalgebraiclens}{\ensuremath{\mathbf{AlgLens}}}_\Psi((A,B), (S,T))
  \coloneqq \C(S, A) \times \D(\Psi S \bullet B, T).
\]
\end{definition}

\begin{proposition}
  Algebraic lenses are mixed optics for the
  actions of the product by the carrier of an algebra
  \((_\U \times) \colon \EM_\Psi \times \C \to \C\) and
  \((_\U \bullet) \colon \EM_\Psi \times \D \to \D\). That is,
  \(\AlgLens_\Psi \cong \Optic_{(_\U \times, _\U \bullet)}\).
\end{proposition}

\begin{proof}
  The \Wt{category} of algebras is
  cartesian, making the forgetful functor \(\U \colon \EM_\Psi \to \C\) monoidal;
  \(\U(\bullet) \times \bullet\) defines a monoidal action.  The forgetful \(\U\) has a left adjoint
  \(\mathcal{F}^\Psi\) such that \(\U \circ \mathcal{F}^\Psi = \Psi\).
\begin{align*}
  & \int^{C \in \mathbf{\EM_\Psi}}
  \C(S, \U  C \times A) \times \D(\U  C \bullet B, T) \\
  \cong & \hintAdjunction{(\times)}{\Delta} \\
  & \int^{C \in \mathbf{\EM_\Psi}}
  \C(S, \U  C) \times \C(S, A) \times \D(\U  C \bullet B, T) \\
  \cong & \hintAdjunction{\mathcal{F}^\Psi}{\U} \\
  & \int^{C \in \mathbf{\EM_\Psi}}
    \EM_\Psi(\mathcal{F}^\Psi S, C) \times \C(S, A) \times
    \D(\U  C \bullet B, T) \\
  \cong & \hintCoyoneda \\
  & \C(S, A) \times \D(\Psi S \bullet B, T). \qedhere
\end{align*}
\end{proof}

\begin{remark}
  Algebraic lenses are a new optic. When \((\bullet)\) is the cartesian product,
  algebraic lenses are given by the usual \texttt{view} function that accesses
  the subfield, and a variant of the \texttt{update} function that takes
  the original source as a computation rather than a value.
\end{remark}

\begin{example}
  The algebraic lens for the list monad \(\mathcal{L} \colon \V \to \V\) is a
  new kind of optic that we dub a \textit{classifying lens}. This is given by two
  functions, the usual \texttt{view} function that accesses the focus, and a
  \texttt{classify} function that takes a list of examples together with some
  piece of data and produces a new example.
  A classifying lens can be trained with a dataset
  (Figure~\ref{fig:irisdataset}) to classify a new focus into a complete data
  structure. A learning algorithm (in this case, a naive version of
  \emph{nearest neighbor}) defines an algebraic lens that can be used to
  classify foci into full data structures.

\begin{figure}[H]
\centering
  \begin{lstlisting}
 let iris =
 [ Iris Setosa 4.9 3.0 1.4 0.2
 , Iris Setosa 4.7 3.2 1.3 0.2
 , ...
 , Iris Virginica 5.9 3.0 5.1 1.8 ]

 measure :: AlgebraicLens [] Measurements Flower
 measure = mkAlgebraicLens @[] measurements learn
  where
   distance :: Measurements -> Measurements -> Float
   distance (Measurements a b c d) (Measurements x y z w) =
      (sqrt . sum . fmap (**2)) [a-x,b-y,c-z,d-w]

   learn :: [Flower] -> Measurements -> Flower
   learn l m = Flower m (species (minimumBy
     (compare `on` (distance m . measurements)) l))

 >>> (iris !! 4) ^. measure
 (5.0, 3.6, 1.4, 0.2)

 >>> iris & measure .? Measurements 4.8 3.1 1.5 0.1
 Iris Versicolor (4.8, 3.1, 1.5, 0.1)
\end{lstlisting}
\caption{Fisher's \texttt{iris} dataset \cite{fisher36},
  samplying lengths and widths of sepals and petals of three species of iris.
  A classifying lens
  (\texttt{measure}) is used both for accessing the measurements of a point in
  the \hask{iris} dataset and to classify \emph{new} measurements (not already in the dataset) into a species
  (\texttt{Versicolor}).}\label{fig:exampleIris}\label{fig:irisdataset}
\end{figure}
\end{example}

\begin{remark}\label{def:achromatic}\label{def:apochromatic}
  The algebraic lens for the \textit{maybe monad}
  \(\mathcal{M} \colon \V \to \V\) was studied in \cite[\S
  5.2]{boisseau17} under the name \textbf{achromatic lens}. It is motivated by
  the fact that, sometimes in practice, lenses come naturally equipped with a
  \texttt{create} function~\cite[\S 3]{foster05} along the usual \texttt{view}
  and \texttt{update}.
For this implementation, we note that
\[\C(S,A) \times \C(\mathcal{M} S \times B,T) \cong \C(S,A) \times \C(S\times B,T) \times \C(B,T).\]
\end{remark}

\begin{remark}
  The name \textit{achromatic lens} can also refer to a different proposal in \cite[\S 4.10]{riley18},
  \[\C(S, [B,T] + 1) \times \C(S,A) \times \C(B,T).\]
  This is the optic for the action \(\odot \colon \C \to [\C,\C]\) defined by
  \(M \odot A \coloneqq (M + 1) \times A\). It is not exactly equivalent to the
  \emph{achromatic lens} in \cite{boisseau17}, which is the optic for the action
  $(\times) \colon \textrm{\(\mathcal{M}\)-Alg} \to [\C,\C]$ defined by
  \(\U M \times A\). It can be implemented by \texttt{view} and \texttt{create}
  functions, this time together with a \texttt{maybeUpdate} that is allowed to
  fail.
\end{remark}

\subsubsection{Lenses in a closed category}
Linear lenses are a different generalization of lenses which relies on a closed
monoidal structure. Their advantage is that we do not need to require our
enriching category to be cartesian anymore.

\begin{definition}[{{\cite[\S 4.8]{riley18}}}]\label{def:linearlens}
  Let \((\D, \otimes, [\ ,\ ])\) be a right closed \Vt{category} with a monoidal
  \Vt{action} \((\bullet) \colon \D \otimes \C \to \C\) to an arbitrary
  \Vt{category} \(\C\). A \textbf{linear lens} is an element of
  \[
    \defining{linklinearlens}{\ensuremath{\mathbf{LinearLens}}}_{\otimes,\bullet} ((A, B), (S, T)) \cong
    \C(S, [B , T] \bullet A).
  \]
\end{definition}

\begin{proposition}
  Linear lenses are mixed optics (as in \Cref{def:optic}) for the
  actions of the monoidal product \((\otimes) \colon \D \times \D \to \D\) and
  \((\bullet) \colon \D \times \C \to \C\). That is,
  \(\LinearLens_{\otimes,\bullet} \cong \Optic_{\bullet, \otimes}\)
\end{proposition}
\begin{proof}
  The monoidal product has a right adjoint given by the exponential.
  \begin{align*}
    & \int^{D \in \D} \C(S,D \bullet A) \otimes \D(D \otimes B, T) \\
    \cong & \hintAdjunction{(- \otimes B)}{[B,-]} \\
    & \int^{D \in \D} \C(S , D \bullet A) \otimes \D(D , [B , T]) \\
    \cong & \hintCoyoneda \\
    & \C(S , [B , T] \bullet A). \qedhere
  \end{align*}
\end{proof}

\subsubsection{Prisms}

Prisms pattern-match on data structures and handle a possible failure to
match. They are given by a \texttt{match} function that tries to access the
matched structure and a \texttt{build} function that constructs an abstract type
from one of its possible matchings.
Prisms happen to be lenses in the opposite category. However, they can also be
described as optics in the original category for a different pair of actions. We
will provide a derivation of prisms that is dual to our derivation of lenses for
a cartesian structure. This derivation specializes to the pair
\(\C(S,T+A) \times \C(B,T)\).
\begin{definition}\label{def:prism}
  Let \(\D\) be a cocartesian \Wt{category} with a monoidal \Wt{action}
  \((\bullet) \colon \D \times \C \to \C\) to an arbitrary category \(\C\). A
  \textbf{prism} is an element of
  \[\defining{linkprism}{\ensuremath{\mathbf{Prism}}}((A, B), (S,T)) \coloneqq \C(S , T \bullet A) \times \D(B , T).
  \]
  In other words, a prism from \((S,T)\) to \((A,B)\) is a lens from \((T,S)\)
  to \((B,A)\) in the opposite categories \(\D^{op}\) and \(\C^{op}\). However,
  they can also be seen as optics from \((A,B)\) to \((S,T)\).
\end{definition}

\begin{proposition}
  Prisms are mixed optics (as in \Cref{def:optic}) for the actions of
  the coproduct \((+) \colon \D \times \D \to \D\) and
  \((\bullet) \colon \C \times \D \to \D\). That is,
  \(\Prism \cong \Optic_{(\bullet,+)}\).
\end{proposition}
\begin{proof}
  The coproduct \((+) \colon \D \times \D \to \D\) is left adjoint to the diagonal
  functor.
  \begin{align*}
  & \int^{D \in \D} \C(S,D \bullet A) \times \D(D + B, T) \\
  \cong & \hintAdjunction{(+)}{\Delta} \\
  & \int^{D \in \D} \C(S , D \bullet A) \times \D(D , T) \times \D(B , T) \\
  \cong & \hintCoyoneda \\
  & \C(S , T \bullet A) \times \D(B , T). \qedhere
  \end{align*}
\end{proof}
\begin{remark}[Prisms in a symmetric monoidal category]
  A prism in a symmetric monoidal category $\C$ is a match and build pair where
  the build is a monoid homomorphism,
  \[\mathbf{MonPrism}((A, B), (S, T)) \coloneqq \C(S,\U T \otimes A) \times \CMon(B,T).\]
\end{remark}

\begin{remark}[Prisms with coalgebraic context]\label{remark:coalgebraicprism}
  Let \(\Theta\) be a \Wt{comonad} in a cocartesian \Wt{category} \(\D\) and
  \((\bullet) \colon \D \times \C \to \C\) a monoidal \Wt{action}. A \textbf{coalgebraic prism} is an element of
  \[\mathbf{AlgPrism}_\Theta((A, B), (S, T))
    \coloneqq \C(S, \Theta T \bullet A) \times \D(B, T).\]
  The coalgebraic
  variant is given by a \texttt{cMatch} function, that captures the failure into
  a comonad, and the usual \texttt{build} function.
\end{remark}

\subsection{Traversals}

Traversals extract the elements of a finitary container \cite{oconnor15} into an ordered list, allowing us
to iterate over the elements without altering the container
(Figure~\ref{fig:exampleTraversal}). Traversals are constructed from a single
\texttt{extract :: s -> ([a], [b] -> t)} function that both outputs the elements
and takes a new list to update them. Usually, we require the length of the input to
the function of type \hask{[b] -> t} to be the same as the length of \hask{[a]}.
This restriction can be also encoded into the type when dependent types are
available.

\begin{figure}[H]
\centering
  \begin{lstlisting}
   let places =
     [ "43 Adlington Rd, Wilmslow, United Kingdom"
     , "26 Westcott Rd, Princeton, USA"
     , "St James's Square, London, United Kingdom" ]

   each :: Traversal a [a]
   each = mkTraversal (\x -> (x , id))

   >>> places & each.asAddress.street %
    [ "43 ADLINGTON RD, Wilmslow, United Kingdom"
    , "26 WESTCOTT RD, Princeton, USA"
    , "ST JAMES'S SQUARE, London, United Kingdom"
    ]
\end{lstlisting}    
\caption{The composition of a traversal (\texttt{each}) with a prism
  (\texttt{asAddress}) and a lens (\texttt{street}) is used to parse a collection of
  strings and modify one of the resulting subfields.}\label{fig:exampleTraversal}
\end{figure}

\begin{definition}
\label{def:traversal}
  Let $\mathbf{C}$ be a symmetric monoidal closed \Vt{category} with countably
  infinite coproducts.
  A \textbf{traversal} is an element of
  \[\defining{linktraversal}{\ensuremath{\mathbf{Traversal}}}((A, B), (S, T)) \coloneqq
    \C\left(S , \int^{n \in \mathbb{N}} A^n \otimes [B^n, T] \right).
  \]
\end{definition}

\begin{remark}
  Let \((\mathbb{N},+)\) be the free strict monoidal \Vt{category} on one object.
  Ends and coends indexed by \(\mathbb{N}\) coincide with products and
  coproducts, respectively. Here \(A^{(-)} \colon \mathbb{N} \to \C\) is the
  unique monoidal \Vt{functor} sending the generator of \(\mathbb{N}\) to
  \(A \in \C\). Each functor \(X \colon \mathbb{N} \to \C\) induces a
  \textit{power
    series} \[\defining{linkpw}{\ensuremath{\operatorname{Pw}}}_X(A) = \int\nolimits^{n \in \mathbb{N}} A^n \otimes X_{n}.\]
  This defines an action
  \(\Pw \colon [\mathbb{N}, \C] \to [\C,\C]\) sending the indexed
  family to its power series\ifdefined\EXTENDEDABSTRACT\else\ (see Remark~\ref{remark:series})\fi. We propose a
  derivation of the \emph{traversal} as the optic for power series.
\end{remark}

\begin{proposition}\label{prop:traversal}
  Traversals are optics (as in \Cref{def:optic}) for power series.
  That is, \(\Traversal \cong \Optic_{\Pw,\Pw}\).
\end{proposition}
\begin{proof}
  The derivation generalizes that of linear lenses (\Cref{def:linearlens}).

\noindent
\resizebox{\linewidth}{!}{
\begin{minipage}{\linewidth}
\begin{align*}
& \int^{X  \in [ \mathbb{N} , \C]}
  \C\left(S , \int^{n \in \mathbb{N}} A^n \otimes X_n\right) \otimes
  \C\left(\int^{n \in \mathbb{N}} B^n \otimes X_n , T\right) \\
\cong & \hintContinuity \\
& \int^{X  \in [ \mathbb{N} , \C]} \C \left(S , \int^{n \in \mathbb{N}}  A^n \otimes X_n \right) \otimes
\int_{n \in \mathbb{N}} \C\left( X_n \otimes B^n , T\right) \\
\cong & \hintAdjunction{(- \otimes B^{n})}{[B^{n},-]} \\
& \int^{X \in [ \mathbb{N} , \C]} \C\left(S , \int^{n \in \mathbb{N}}  A^n \otimes X_n \right) \otimes
\int_{n \in \mathbb{N}} \C\left( X_n , [B^n , T]\right) \\
\cong & \hintNaturalTransformation \\
& \int^{X \in [ \mathbb{N} , \C]} \C\left(S , \int^{n \in \mathbb{N}} A^n \otimes X_n  \right) \otimes
  [ \mathbb{N} , \C ] \left(X_{(-)} , [B^{(-)} , T]\right) \\
\cong & \hintCoyoneda \\
& \C \left(S , \int^{n \in \mathbb{N}} A^n \otimes [B^n, T] \right). \qedhere
\end{align*}
\end{minipage}}
\end{proof}

The derivation from the general definition of optic to the concrete description
of lenses and prisms in functional programming was first described by
\cite{pickering17} and \cite{milewski17}, but finding a derivation of
traversals like the one presented here, fitting the same elementary pattern as
lenses or prisms, was left as an open problem. It should be noted, however, that
derivations of the traversal as the optic for a certain kind of functor called
\textbf{Traversables} (which should not be confused with traversals themselves)
have been previously described by \cite{boisseau18,riley18}. For a derivation
using Yoneda, \cite{riley18} recalls a parameterised adjunction that has an
equational proof in the work of \cite{jaskelioff15}. These two derivations do
not contradict each other: two different classes of functors can generate the
same optic; if, for instance, the adjunction describing both of them gives rise
to the same monad. This seems to be the case here: traversables are coalgebras
for a comonad and power series are the cofree coalgebras for the same
comonad~\cite[\S 4.2]{roman19}.

In the \(\mathbf{Sets}\mbox{-based}\) case, the relation between traversable functors, applicative functors~\cite{mcbride08} and these power series functors has been studied by \cite{rypacek12}.

\begin{remark}\label{remark:series}
  The \Vt{functor} \(\Pw \colon [\mathbb{N},\C] \to [\C,\C]\) is
  actually a monoidal action thanks to the fact that two power series
  functors compose into a power series functor.
\begin{align*}
  &\int^{m \in \mathbb{N}} \left( \int^{n \in \mathbb{N}} A^n \otimes C_n \right)^m \otimes D_m \\
  \cong & \hint{Product distributes over colimits} \\
  &\int^{m} \int^{n_1,\dots,n_m} A^{n_1} \otimes \dots \otimes A^{n_m}
   \otimes C_{n_1} \otimes \dots \otimes C_{n_m} \otimes D_m \\
  \cong & \hint{Rearranging terms} \\
  &\int^{k \in \mathbb{N}} A^k \otimes \left( \sum_{n_1 + \dots + n_m = k}
  C_{n_1} \otimes \dots \otimes C_{n_m} \otimes D_m \right).
\end{align*}
We are then considering an implicit non-symmetric monoidal structure where the
monoidal product
\((\bullet) \colon [\mathbb{N},\C] \otimes [\mathbb{N},\C] \to [\mathbb{N},\C]\)
of \(C_n\) and \(D_n\) can be written as follows; and the relevant copowers exist because
of $\C$ having an initial object.
\[\int^{m} \int^{n_1,\dots , n_m}
  \mathbb{N}(n_1 + \dots + n_m, k) \cdot
  C_{n_1} \otimes \dots \otimes C_{n_m} \otimes D_m.
\]
This is precisely the structure described by Kelly \cite[\S 8]{kelly:operads05} for the
study of non-symmetric operads. A similar monoidal structure is described there
when we substitute \(\mathbb{N}\) by \((\mathbb{P},+)\), the
\Vt{\textit{category of permutations}} defined as the free strict symmetric
monoidal category on one object. The same derivation can be repeated with this
new structure to obtain an optic similar to the traversal, with the difference that
elements are not ordered explicitly, and given by
\[\C \left(S , \int^{n \in \mathbb{P}} A^n \otimes [B^n, T] \right).\]
\end{remark}
 
\subsubsection{Affine traversals}

\emph{Affine traversals} strictly generalize prisms and linear lenses in the
non-mixed case allowing a \textit{lens-like} accessing pattern to fail; they are
used to give a concrete representation of the composition between a lens and a
prism. An affine traversal is implemented by a single \texttt{access} function.

\begin{definition}\label{def:affine}
  Let \(\W\) be cartesian closed and let \(\C\) be a symmetric monoidal closed
  \Wt{category} that is also cocartesian. An \textbf{affine traversal} is an
  element of
  \[\defining{linkaffine}{\ensuremath{\mathbf{Affine}}}_{\otimes}\left( (A, B), (S, T) \right)
  := \mathbf{C}(S , T + A \otimes [B,T]).\]
\end{definition}

\begin{proposition}
  Affine traversals are optics (as in \Cref{def:lens}) for the action
  \((+\otimes) \colon \C^2 \to [\C,\C]\) that sends \(C, D \in \C\) to the
  functor \(C + D \otimes (-)\). That is,
  \(\Affine_{\otimes} \cong \Optic_{(+\otimes),(+\otimes)}\).
\end{proposition}
\begin{proof}
  The action uses that the monoidal product, which is in this case a left
  adjoint, distributes over the coproduct.
  \begin{align*}
  & \int^{C,D} \C(S , C + D \otimes A) \times \C(C + D \otimes B, T) \\
   \cong & \hintAdjunction{(+)}{\Delta} \\
  & \int^{C,D} \C(S , C + D \otimes A) \times \C(C , T) \times \C(D \otimes B, T) \\
    \cong & \hintCoyoneda
            \hintAdjunction{(- \otimes B)}{[B,-]} \\
  & \int^D \C(S , T + D \otimes A) \times \C(D , [B, T]) \\
  \cong & \hintCoyoneda \\
  & \C(S , T + A \otimes [B,T]). \qedhere
  \end{align*}
\end{proof}

\subsubsection{Kaleidoscopes}

\emph{Applicative} functors are commonly used in functional programming as
they provide a convenient generalization to monads with better properties for
composition; they form the \Vt{category} \(\defining{linkapp}{\ensuremath{\mathbf{App}}}\) of monoids with
respect to Day convolution \((\ast) \colon [\V,\V] \times [\V,\V] \to [\V,\V]\),
which is defined as
\[
  (F \ast G)(A) = \int^{X,Y \in \V} \V(X \otimes Y, A) \otimes FX \otimes GY.
\]
Alternatively, they are lax monoidal \Vt{functors} for the cartesian structure.
It is natural to ask what is the optic associated with applicative functors. We know
from a basic result in category theory~\cite[\S VII, Theorem 2]{maclane78} that,
as the category \([\V,\V]\) has coproducts indexed by the natural numbers, and Day convolution
distributes over them, the free applicative functor can be computed as the
colimit \((I + F + F^{\ast 2} + F^{\ast 3} + \dots)\). Having characterized free
applicative functors, computing their associated optic amounts to an application
of coend calculus.

\begin{definition}\label{def:kaleidoscope}
  A \textbf{kaleidoscope} is an element of
\[\mathbf{Kaleidoscope}\left((A, B), (S, T)\right) :=
  \prod_{n \in \mathbb{N}} \V \left([A^n, B] , [S^n, T] \right).
\]
\end{definition}

\begin{proposition}
  Kaleidoscopes are optics for the action of applicative functors.
\end{proposition}
\begin{proof}
  Let \(\U \colon \App \to [\V,\V]\) be the forgetful functor from the category
  of applicatives.
\begin{align*}
& \int^{F \in \App} \V(S, \U FA) \otimes \V(\U FB,T)
\\ \cong & \hintYoneda \\
& \int^{F \in \App}  \V\left(S, \int_{C \in \V} \big[ [A,C],\U FC \big] \right) \otimes \V(\U FB,T)
\\ \cong & \hintContinuity \\
& \int^{F \in \App} \int_{C} \V\left(S, \big[ [A,C],\U FC\big] \right) \otimes \V(\U FB,T)
\\ \cong & \hintAdjunction{([A,C] \otimes -)}{\big[ [A,C],-\big]} \\
& \int^{F \in \App} \left( \int_{C} \V\left([A,C] \otimes S, \U FC \right) \right) \otimes \V(\U FB,T)
\\ \cong & \hintNaturalTransformation \\
& \int^{F \in \App} [\V, \V]\left([A,-] \otimes S, \U F \right) \otimes \V(\U FB,T)
\\ \cong & \hint{Adjunction of free applicatives} \\
& \int^{F \in \App} \App\left(\sum_{n\in \mathbb{N}} [A^n , -] \otimes S^n , F \right) \otimes \V(\U FB,T)
\\ \cong & \hintCoyoneda \\
& \V \left(\sum_{n \in \mathbb{N}} S^n \otimes [A^n , B],T \right)
\\ \cong & \hintContinuity \hintAdjunction{(S^{n} \otimes -)}{[S^{n},-]} \\
& \prod_{n \in \mathbb{N}} \V \left([A^n , B],  [S^n, T] \right). \qedhere
\end{align*}
\end{proof}

\begin{remark}
  The free applicative we construct here is known in Haskell programming as the
  \texttt{FunList} applicative \cite{laarhoven09b}.
  In the same way that traversables are written in
  terms of lists, we can approximate Kaleidoscopes as a single function
  \texttt{aggregate :: ([a] -> b) -> ([s] -> t)} that takes a folding for the foci and outputs a folding for
  the complete data structure. Kaleidoscopes are a new optic, and we propose to
  use them as accessors for pointwise foldable data structures.
\end{remark}

Kaleidoscopes pose a problem on the side of applications: they cannot be
composed with lenses to produce new kaleidoscopes. This is because
the constraint defining them (\hask{Applicative}) is not a
superclass of the constraint defining lenses: a functor given by a product is
not applicative, in general. However, a functor given by a product \emph{by a
  monoid} is applicative. This means that applicatives can be composed with
lenses whose residual is a monoid, which are precisely the newly defined
\emph{classifying lenses} (\Cref{def:algebraiclens}).

\begin{figure}[H]
\centering
  \begin{lstlisting}
 representative :: Kaleidoscope Float Measurements
 representative = mkKaleidoscope aggregate
  where
   aggregate f l = Measurements
    (f (fmap sepalLe l)) (f (fmap sepalWi l))
    (f (fmap petalLe l)) (f (fmap petalWi l))

 iris & measure.representative >- mean
  >>> Iris Versicolor; Sepal (5.843, 3.054); Petal (3.758, 1.198)
\end{lstlisting}
\caption{Following the previous Example~\ref{fig:exampleIris}, a kaleidoscope (\texttt{representative}) is composed with an algebraic lens to create a new point in the dataset by aggregating measurements with some function (\texttt{mean}, \texttt{maximum}) and then classifying it.}\label{fig:exampleKaleidoscope}
\end{figure}

\subsection{Grates}

\emph{Grates} create a new structure when provided with a way of creating a new focus from
a view function. They are given by a single \texttt{grate} function with the form
of a nested continuation.

\begin{definition}
\label{def:grate}
  Let \(\C\) be a symmetric closed \Vt{category} (as in~\cite{day78}). Let
  \((\bullet) \colon \C \times \D \to \D\) be an arbitrary action. A
  \textbf{grate} is an element of
\[\defining{linkgrate}{\ensuremath{\mathbf{Grate}}}\left((A, B), (S, T)\right) := \D([S,A] \bullet B, T).\]
\end{definition}
\begin{proposition}
  Grates are mixed optics for the action of the exponential and
  \((\bullet) \colon \C \times \D \to \D\). That is,
  \(\Grate \cong \Optic_{[\ ,\ ],\bullet}\).
\end{proposition}
\begin{proof}
  We know from \cite[Proposition 1.2]{day78} that the adjunction
  $[-,A]^{op} \dashv [-,A]$ holds in any symmetric monoidal category.
\begin{align*}
& \int^{C \in \C} \C(S,[C,A]) \otimes \D(C \bullet B, T) \\
\cong & \hintAdjunction{[-,A]^{op}}{[-,A]} \\
& \int^{C \in \C} \C(C,[S, A]) \otimes \D(C \bullet B, T) \\
\cong & \hintCoyoneda \\
& \D([S, A] \bullet B, T) \qedhere
\end{align*}
\end{proof}

The description of a grate and its profunctor representation in terms of
``closed'' profunctors was first reported by Deikun and
O'Connor~\cite{deikun15}; it can be seen as a consequence of the profunctor
representation theorem (\Cref{th:profrep}).

\subsubsection{Glasses}

\emph{Glasses} are a new optic that strictly generalizes grates and lenses for
the cartesian case. In functional programming, glasses can be implemented by a
single function \texttt{glass} that takes a way of transforming \emph{views} into
new foci and uses it to update the data structure. We propose to use them as a
concrete representation of the composition of lenses and grates.

\begin{definition}
  \label{def:glass}
  A \textbf{glass} in a cartesian closed \Wt{category} is an element of
\[\defining{linkglass}{\ensuremath{\mathbf{Glass}}} \left( (A, B), (S, T) \right) := \C(S \times [[S, A], B] , T).\]
\end{definition}

\begin{proposition}
  Glasses are optics for the action
  \((\bullet \times[\bullet, -]) \colon \C^{op} \times \C \to [\C,\C]\) that sends \(C, D \in \C\)
  to the \Wt{functor} \(D \times [C,-]\).
\end{proposition}
\begin{proof}
\begin{align*}
& \int^{C,D} \C(S, C \times [D,A]) \times \C([D,B] \times C , T) \\
\cong & \hintAdjunction{\Delta}{(\times)} \\
& \int^{C,D} \C(S, C) \times \C(S, [D,A]) \times \C([D,B] \times C , T) \\
\cong & \hintCoyoneda \\
& \int^D \C(S, [D, A]) \times \C([D,B] \times S, T) \\
\cong & \hintAdjunction{[-,A]^{op}}{[-,A]} \\
& \int^D \C(D, [S , A]) \times \C([D, B] \times S, T) \\
\cong & \hintCoyoneda \\
& \C([[S,A], B] \times S , T). \qedhere
\end{align*}
\end{proof}

\subsection{Getters, reviews and folds}

Some constructions, such as plain morphisms, can be regarded as degenerate
cases of optics. We will describe these constructions (\emph{getters},
\emph{reviews} and \emph{folds}~\cite{kmett15}) as mixed optics. All of them set
one of the two base categories to be the terminal category and, contrary to most
optics, they act only unidirectionally. We will derive the usual
implementations of getters as \texttt{s -> a}, of reviews as \texttt{b -> t}, and
of folds as \texttt{s -> [a]}. Their profunctor representation can be seen as a
covariant or contravariant application of the Yoneda lemma.

\begin{definition}\label{def:getter}\label{def:review}
  Let \(\C\) be an arbitrary \Vt{category}. Getters (morphisms $\C(S,A)$) are
  degenerate optics for the trivial action on the covariant side. Reviews (morphisms
  $\C(B,T)$) are degenerate optics for the trivial action on the contravariant
  side. In other words, we can obtain plain morphisms as a particular case of optic when
  \(\M = \mathbf{1}\) and \(\D = \mathbf{1}\) or \(\C = \mathbf{1}\),
  respectively.
\end{definition}

The category of \textbf{foldable} functors, $\mathbf{Foldable}$, is the slice category on the list
functor \(\defining{linklist}{\ensuremath{\mathcal{L}}} \colon \V \to \V\), also known as the free monoid functor.
Using the fact that \(\ListMonad\) is a monad, the slice category
\([\V,\V]/ \ListMonad\) can be made monoidal in such a way that the forgetful functor
\([\V,\V]/ \ListMonad \to [\V,\V]\) becomes monoidal.

\begin{definition}\label{def:fold}
  Folds are optics for the action of foldable functors and the trivial
  action on the contravariant side.
\[\mathbf{Fold}((A,\ast), (S,\ast)) = \int^{F \in \mathbf{Foldable} } \V(S , F A) \cong \V(S, \ListMonad A).\]
Folds admit a concrete description, $\V(S, \ListMonad A)$, which can be reached from
the definition of coends as colimits. A coend from a diagram category with a
terminal object is determined by the value at the terminal object and here, $\ListMonad$ is the terminal object. The same
technique can be used to prove that the optic for the slice category over a
monad \(\mathcal{G} \colon \V \to \V\) has a concrete form given by
\(\C(S, \mathcal{G}A)\).
\end{definition}
 
\subsection{Setters and adapters}

We finish our examples of optics with two extremes: the actions from the initial and terminal actions. These are the identity action \(\mathrm{id} \colon [\V , \V] \to [\V , \V]\) and picking the monoidal unit, $\mathrm{I} \colon \mathbf
{1} \to [\V, \V ]$.

\begin{definition}\label{def:setter}
  A setter~\cite{kmett15} is an element of
  \[\defining{linksetter}{\ensuremath{\mathbf{Setter}}} ((A, B), (S, T)) \coloneqq \V([A,B], [S,T]).\]
\end{definition}
\begin{proposition}
  Setters are optics for identity action
  \(\mathrm{id} \colon [\V , \V] \to [\V , \V]\). That is,
  \(\Setter \cong \Optic_{\mathrm{id},\mathrm{id}}\).
\end{proposition}
\begin{proof}
  The concrete derivation is described by Riley \cite[\S 4.5.2]{riley18}. There, it is required that functors are strong; this requirement holds automatically for every enriched functor in the base of enrichment \cite{levy2019strong}.
  \begin{align*}
  & \int^{F \in [ \V , \V ]} \V(S, FA) \otimes \V(FB,T) \\
  \cong & \hintYoneda \\
  & \int^{F \in [ \V , \V ]} \left( \int_{C \in \V} [\V(A,C), \V(S, FC)] \right) \otimes \V(FB,T) \\
  \cong & \hint{Copower} \\
  & \int^{F \in [ \V , \V ]} \left( \int_{C \in \V} \V(S \otimes [A,C], FC) \right) \otimes \V(FB,T) \\
  \cong & \hintNaturalTransformation \\
  & \int^{F \in [ \V , \V ]} [\V,\V](S \otimes [A,-], F) \otimes \V(FB,T) \\
  \cong & \hintCoyoneda \\
  & \V(S \otimes [A,B], T) \\
  \cong & \hintAdjunction{(S \otimes-)}{[S,-]} \\
 &  \V([A,B], [S,T]). \qedhere
  \end{align*}
\end{proof}

\begin{remark}
  In functional programming, we implicitly restrict to the case where \(\V\) is
  a cartesian category, and we curry this description to obtain the usual
  representation of setters as a single \texttt{over} function.
\end{remark}

\begin{definition}\label{def:adapter}
  Adapters~\cite{kmett15} are morphisms in \(\C^{op} \otimes \D\). They
  are optics for the action $\mathrm{Id} \colon \mathbf{1} \to [\C,\C]$.
  \[\defining{linkadapter}{\ensuremath{\mathbf{Adapter}}} ((A, B), (S, T)) \coloneqq \C(S,A) \otimes \D(B,T).\]
\end{definition}

\subsection{Optics for (co)free}\label{sec:opticsforcofree}
A common pattern that appears across many optic derivations is that of computing the optic associated with a class of functors using an adjunction to allow for an application of the Yoneda lemma. This observation can be generalized to a class of concrete optics.

Consider some \Vt{endofunctor} \(H \colon \V \to \V\). Any objects \(A,B,S,T \in \V\) can be regarded as functors from the unit \Vt{category}. The following isomorphisms are the consequence of the fact that left and right global Kan extensions are left and right adjoints to precomposition, respectively.
\begin{align*}
\V(S,HA) \cong [\V,\V]( \mathsf{Lan}_A S,H),\\
\V(HB,T) \cong [\V,\V](H, \mathsf{Ran}_B T).
\end{align*}
These extensions exist in \(\V\) and they are given by the formulas
\[
  \mathsf{Lan}_AS \cong [A,-] \otimes S,
  \quad
  \mathsf{Ran}_BT \cong \big[ [-,B],T\big].
\]

\begin{proposition}[{{\cite[\S 3.4.7]{roman19}}}]
\label{prop:opticfree}
  Let the monoidal \Vt{action} \(\U  \colon \M \to [\V,\V]\) have a left adjoint \(\mathcal{L} \colon [\V,\V] \to \M\), or, dually, let it have a right adjoint \(\mathcal{R} \colon [\V,\V] \to \M\). In both of these cases the optic determined by that monoidal action has a concrete form, given by
  \[\V(\mathcal{UL}([A,-] \otimes S)(B),T) \quad\mbox{ or }\quad \V(S,\mathcal{UR} \big[ [-,B],T\big](A)),\]
  respectively.
\end{proposition}
\begin{proof}
We prove the first case. The second one is dual.
\begin{align*}
& \int^{M \in \M} \V(S, \U MA) \otimes \V(\U MB,T)
\\ \cong & \hint{Kan extension} \\
& \int^{M \in \M} [\V, \V]\left( \mathsf{Lan}_A S , \U M \right) \otimes \V(\U MB,T)
\\ \cong & \hintAdjunction{\mathcal{L}}{\mathcal{U}} \\
& \int^{M \in \M} \M\left(\mathcal{L} \mathsf{Lan}_{A}S , M \right) \otimes \V(\U MB,T)
\\ \cong & \hintCoyoneda \\
& \V \left( \mathcal{UL} \mathsf{Lan}_{A}S (B) , T \right). & \qedhere
\end{align*}
\end{proof}

\section{Tambara theory}
\label{sec:tambaratheory}

A fundamental feature of optics is that they can be represented as a single
polymorphic function. Optics in this form are called \emph{profunctor optics}
and we say this function is their \emph{profunctor representation}. Profunctor
optics can be easily composed, even if they are from different families:
composition of optics as polymorphic functions becomes ordinary function
composition. Figure~\ref{fig:lensandprism} shows an example of this phenomenon.
Profunctor optics are functions polymorphic over profunctors endowed with some
extra algebraic structure. This extra structure depends on the family of the
optic they represent. For instance, \emph{lenses} are represented by functions
polymorphic over \emph{cartesian} profunctors, while \emph{prisms} are
represented by functions polymorphic over \emph{cocartesian}
profunctors~\cite[\S 3]{pickering17}. Milewski \cite{milewski17} notes that the algebraic
structures accompanying these profunctors are precisely \emph{Tambara
  modules}, a particular kind of profunctor that has been used
to characterize the monoidal centre of convolution monoidal
categories~\cite{tambara06}. Because of this correspondence, categories of
lenses or prisms can be obtained as particular cases of the ``Doubles for
monoidal categories'' construction defined by Pastro and Street~\cite[\S 6]{pastro08}.
The \emph{double}\footnote{Double, as used by Pastro and Street \cite{pastro08},
  should not be confused with \emph{double} in the sense of \emph{double
    category}.} of an arbitrary monoidal \Vt{category} \((\AA,\otimes,I)\), is a
promonoidal
\footnote{A promonoidal category \(\AA\) is a generalization of a monoidal category with functors replaced by profunctors. For instance, a tensor product is a profunctor \( P \colon \AA^{op} \otimes \AA \otimes \AA \to \V \) and the unit is \(J \colon \AA^{op} \to \V\)}
\Vt{category} \(\DA\) whose hom-objects are defined as
\[
\DA((A,B),(S,T)) \coloneqq\int^{C \in \AA} \AA(S,C \otimes A) \otimes \AA(C \otimes B, T).
\]
In the particular case where \(\AA\) is cartesian or cocartesian, the
\Vt{category} \(\DA\) is precisely the category of lenses or prisms over
\(\AA\), respectively. Moreover, one of the main results of \cite[Proposition
6.1]{pastro08} declares that the \Vt{category} \([\mathcal{D}\C, \V]\) of copresheaves over these
\Vt{categories} is equivalent to the \Vt{category} of
Tambara modules on \(\C\). In the case of lenses or prisms, these Tambara
modules are precisely cartesian and cocartesian profunctors, and this correspondence justifies
their profunctor representation.

We will see how the results of Pastro and Street can be directly applied to the
theory of optics, although they are not general enough for our purposes.
Milewski \cite{milewski17} already proposed a unified description of optics, later
extended by \cite{boisseau18} and \cite{riley18}, that requires a
generalization of the original result by Pastro and Street from monoidal
products to arbitrary monoidal actions. In order to capture \Vt{enriched} mixed
optics, we need to go even further and generalize the definition of Tambara
module in two directions. The monoidal category \(\AA\) in their definition
needs to be substituted by a pair of arbitrary categories \(\C\) and \(\D\), and
the monoidal product \(\otimes_\AA \colon \AA \otimes \AA \to \AA\) needs to be
substituted by a pair of arbitrary monoidal actions
\(\actL{}{} \colon \M \otimes \C \to \C\) and
\(\actR{}{} \colon \M \otimes \D \to \D\), from a common monoidal category
\(\M\).

This section can be seen both as a partial exposition of one of the main results
of the work of Pastro and Street~\cite[Proposition 6.1]{pastro08} and a
generalization of their definitions and propositions to the setting that is most
useful for applications in functional programming.
 
\subsection{Generalized Tambara modules}
Originally, \emph{Tambara modules}~\cite{tambara06} were conceived of as a
structure on top of certain profunctors that made them play nicely with
some monoidal action in both their covariant and contravariant components.

For our purposes, Tambara modules represent the different ways in which we can
use an optic. If a profunctor $P$ has Tambara module structure for the monoidal
actions defining an optic, we can use that optic to lift the profunctor applied
to the foci, $P(A,B)$, to the full structures, $P(S,T)$. For instance, the
profunctor $(-) \times B \to (-)$ can be used to lift the projection
$A \times B \to B$ into the \texttt{update} function $S \times B \to T$. In other
words, this profunctor is a Tambara module compatible with all the families of
optics that admit an \texttt{update} function, such as \emph{lenses}. In
programming libraries, that profunctor can be used to define a polymorphic \texttt{update}
combinator that works across different families of optics.

Formally, we want to prove that Tambara modules for the actions \(\actL{}{}\)
and \(\actR{}{}\) are copresheaves over the category
\(\Optic_{\actL{}{},\actR{}{}}\). This will also justify the profunctor
representation of optics in terms of Tambara modules (Theorem~\ref{th:profrep}).

\begin{definition}\label{def:tambara}
  Let \((\M,\otimes,I)\) be a monoidal \Vt{category} with two monoidal actions
  \((\actL{}{}) \colon \M \otimes \C \to \C\) and
  \((\actR{}{}) \colon \M \otimes \D \to \D\). A generalized \textbf{Tambara
    module} consists of a \Vt{profunctor}
  \(P \colon \C^{op} \otimes \mathbf{D} \to \V\) endowed with a family of
  morphisms
  \[\alpha_{M,A,B} \colon P(A,B) \to P(\actL{M}{A} , \actR{M}{B})\]
  \Vt{natural} in \(A \in \C\) and \(B \in \D\) and \Vt{dinatural} in \(M \in \M\),
  which additionally satisfies the two equations
  \begin{align*}
    P(\phi_{A}^L,\phi^{-1R}_{B}) \circ \alpha_{I,A,B}  &= \id, \\
    P(\phi_{M,N,A},\phi^{-1}_{M,N,B}) \circ \alpha_{M \otimes N,A,B}  &= \alpha_{N,A,B} \circ
  \alpha_{M,\actL{N}{A},\actR{N}{B}},
  \end{align*}
  for every \(M,N \in \M\), every \(A \in \C\) and every \(B \in \D\).
 Let us repeat the equations in terms of commutative diagrams for clarity.
  This is a family of morphisms making the following two diagrams commute.
  
  \[\begin{tikzcd}[column sep=4em]
  P(A,B) \rar{\alpha_{M \otimes N,A,B}}\dar[swap]{\alpha_{N,A,B}} &
  P(\actL{(M \otimes N)}{A},\actR{(M \otimes N)}{B}) \dar{P(\phi_{M,N,A}, \varphi^{-1}_{M,N,B})} \\
  P(\actL{N}{A}, \actR{N}{B}) \rar[swap]{\alpha_{M,\actL{N}{A},\actR{N}{B}}} &
  P(\actL{M}{\actL{N}{A}}, \actR{M}{\actR{N}{B}}) \\
  P(A,B) \drar[swap]{\id} \rar{\alpha_{I,A,B}} &
  P(\actL{I}{A}, \actR{I}{B}) \dar{P(\phi_{A},\varphi^{-1}_{B})} \\
  & P(A,B)
  \end{tikzcd}\]
\end{definition}

When \(\V = \Sets\), we can define a morphism of Tambara modules as a natural
transformation \(\eta_{A,B} \colon P(A,B) \to Q(A,B)\)
satisfying \[\eta_{\actL{M}{A},\actR{M}{B}} \circ \alpha_{M,A,B} = \alpha'_{M,A,B}\circ \eta_{A, B}.\]
For an arbitrary \(\V\), Tambara modules are the objects of a \Vt{category}
\(\defining{linktamb}{\ensuremath{\mathbf{Tamb}}}\) whose hom-object from \((P,\alpha)\) to \((Q,\alpha')\) is
computed as the equalizer of the arrows
\begin{align*}
  \int_{X,Y} \V(P(X,Y),Q(X,Y)&) \\
  \xrightarrow{\phantom{AAAAAAA}\pi_{A,B}\phantom{AAAAAAA}}\ & \V(P(A,B), Q(A,B)) \\
  \xrightarrow{\phantom{AAAAA}\V(\id,\alpha'_{A,B})\phantom{AAAAA}}\ & \V(P(A,B), Q(\actL{M}{A},\actR{M}{B}))
\end{align*}
and
\begin{align*}
  \int_{X,Y} \V(P(X,Y&),Q(X,Y)) \\
  \xrightarrow{\phantom{AA}\pi_{\actL{M}{A},\actR{M}{B}}\phantom{AA}}\ &
  \V(P(\actL{M}{A},\actR{M}{B}), Q(\actL{M}{A},\actR{M}{B})) \\
  \xrightarrow{\phantom{xA}\V(\alpha_{A,B},\id)\phantom{Ax}}\ & \V(P(A,B), Q(\actL{M}{A},\actR{M}{B}))
\end{align*}
for each \(A \in \C\), \(B \in \D\) and \(M \in \M\)~\cite[\S 3.2]{pastro08}.

\begin{remark}
  Pastro and Street \cite{pastro08} follow the convention of omitting the
  unitors and the associators of the monoidal category when defining Tambara
  modules. These appear in \Cref{def:tambara} replaced by the
  structure isomorphisms of the two monoidal actions.
\end{remark}
 
\subsection{Pastro-Street's ``double'' comonad}
Tambara modules are coalgebras for a particular comonad~\cite[\S 5]{pastro08} in profunctors.
That comonad has a left adjoint that must therefore be a monad, and then Tambara
modules can be equivalently described as algebras for that monad. We will
describe the \Vt{category} of \emph{generalized} Tambara modules \(\Tamb\) as an
Eilenberg-Moore category first for a comonad and then for its left adjoint
monad. This will be the main lemma (\Cref{lemma:tambaracoalg}) towards the
profunctor representation theorem (\Cref{th:profrep}).

\begin{remark}\label{remark:axiomsmonoidal}
  Before the definition, let us recall the axioms for a monoidal \Vt{action}
  \(\oslash \colon \M \otimes \C \to \C\) of a monoidal \Vt{category}
  \((\M,\otimes,I)\) with coherence isomorphisms \((a,\lambda,\rho)\) to an
  arbitrary category \(\C\). Let
  \[\begin{aligned}
  & \phi_{A} \colon A \cong I \oslash A, && \phi_{M,N,A} \colon M \oslash (N \oslash A) \cong (M \otimes N) \oslash A, \\
\end{aligned}\] be the structure \Vt{natural} isomorphisms of the monoidal action. Note that the following are precisely the axioms for a strong
monoidal functor \(\M \to [\C,\C]\) written as \(\M \otimes \C \to \C\); they
are simplified by the fact that \([\C,\C]\) is strict.
\[\begin{tikzcd}
    I \oslash M \oslash A \dar[swap]{\phi_{I,M,A}} &
    M \oslash A \lar[swap]{\phi_{M \oslash A}} \rar{M \oslash \phi_A} \dar{\id} &
    M \oslash I \oslash A \dar{\phi_{M,I,A}} \\
    (I \otimes M) \oslash A \rar[swap]{\lambda_M \oslash A}&
    M \oslash A &
    (M \otimes I) \oslash A \lar{\rho_M \oslash A}
\end{tikzcd}\]
\[\begin{tikzcd}
    M \oslash N \oslash K \oslash A \rar{\id} \dar[swap]{\phi_{M,N,K\oslash A}} &
    M \oslash N \oslash K \oslash A \dar{M \oslash \phi_{N,K,A}} \\
    (M \otimes N) \oslash K \oslash A \dar[swap]{\phi_{M \otimes N,K,A}} &
    M \oslash (N \otimes K) \oslash A \dar{\phi_{M,N \otimes K,A}} \\
    ((M \otimes N) \otimes K) \oslash A \rar[swap]{a_{M,N,K} \oslash A} &
    (M \otimes (N \otimes K)) \oslash A
\end{tikzcd}\]
\end{remark}

\begin{definition}\label{def:pastrofunctor}
  We start by constructing the underlying \Vt{functor} of the comonad
  \(\Theta \colon \Prof(\C,\D) \to \Prof(\C,\D)\). Consider first the
  \Vt{functor}
  \[T \colon \M^{op} \otimes \M \otimes \C^{op} \otimes \D \otimes \mathbf{Prof}(\C,\D) \to \V, \]
  given by the composition of the actions \((\actL{}{})\) and \((\actR{}{})\)
  with the evaluation \Vt{functor}
  \(\C^{op} \otimes \D \otimes \Prof(\C,\D) \to \V\). On objects, this is given
  by
  \[ T(M,N,A,B,P) \coloneqq P(\actL{M}{A},\actR{N}{B}).
  \]
  By the universal property of the end, this induces a \Vt{functor}
  \(\C^{op} \otimes \D \otimes \Prof(\C,\D) \to \V\) given by
  \[ T'(A,B,P) = \int_{M \in \M} P(\actL{M}{A},\actR{M}{B}),
  \]
  which can be curried into the \Vt{functor}
  \[\Theta P(A,B) \coloneqq \int_{M \in \M} P(\actL{M}{A}, \actR{M}{B}).
  \]
\end{definition}

\begin{proposition}\label{prop:pastrocomonad}
  The \Vt{functor} \(\Theta\) can be endowed with a comonad structure. Its
  counit is the \Vt{natural} family of morphisms
  \(\varepsilon_{P,A,B} \colon \Theta P (A,B) \to P(A,B)\) obtained by
  projecting on the monoidal unit and applying the unitors,
  \[\begin{tikzcd}[column sep=4em]
    \Theta P(A,B) \rar{\pi_{I}} &
    P(\actL{I}{A}, \actR{I}{B}) \rar{P(\phi_{A}, \phi^{-1R}_{B})} & P(A,B).
  \end{tikzcd}\]
  Its comultiplication is given by
  \(\delta_{P,A,B} \colon \Theta P(A,B) \to \Theta \Theta P(A,B)\), the
  \Vt{natural} family of transformations obtained as the unique morphisms
  factorizing
  \[\begin{tikzcd}[column sep=1.5em]
  \displaystyle \Theta P (A,B)
  \ar{r}{P\left(\phi_{M,N,A},\phi^{-1R}_{M,N,B}\right) \circ \pi_{M \otimes N} }
  &[18ex] P(\actL{M}{\actL{N}{A}}, \actR{M}{\actR{N}{B}})
  \end{tikzcd}\]
  through the projection
  \[\begin{tikzcd}[column sep=4em]
  \Theta \Theta P (A,B)
  \rar{\pi_M \circ \pi_N} &
  P(\actL{M}{\actL{N}{A}} , \actR{M}{\actR{N}{B}})
  \end{tikzcd}\]
  for every \(M,N \in \M\).
\end{proposition}

\begin{proof}
  In order to keep the diagrams in this proof manageable, we introduce the notation
  \(P[M](A,B) \coloneqq P(\actL{M}{A},\actR{M}{B})\). We will show that this
  construction satisfies the comonad axioms.

  We first prove left counitality,
  \(\Theta \varepsilon_P \circ \delta_P = \mathrm{id}_{\Theta P}\), which
  follows from the commutativity of the following diagram.
\[\begin{tikzcd}
    \Theta P(A,B) \rar{\delta_P} \ar{dd}[swap]{\pi_{I \otimes M}} &
    \Theta\Theta P(A,B) \rar{(\Theta \varepsilon)_P} \dar{\pi_M}  &
    \Theta P(A,B) \dar{\pi_M} \\
    & \Theta P[M](A,B) \dar{\pi_I} \rar{\varepsilon_{P[M]}} &
    P[M](A,B) \dar{\id} \\
    P[I \otimes M](A,B) \ar[swap]{r}[yshift=-1ex]{P(\phi_{I,M,A},\phi^{-1R}_{I,M,B})} &
    P[M][I](A,B) \ar[swap]{r}[yshift=-1ex]{P(\phi_{\actL{M}{A}},\phi^{-1R}_{\actR{M}{B}})} &
    P[M](A,B)
  \end{tikzcd}\]
The left pentagon of this diagram commutes because of the
definition of \(\delta\). The upper right square commutes because of
functoriality of ends and naturality of \(\pi_M\). The lower right square
commutes because of the definition of \(\varepsilon\). By the axioms of the
monoidal actions (Remark~\ref{remark:axiomsmonoidal}), the bottom edge of the
whole diagram can be rewritten as
\begin{align*}
 P(\phi_{\actL{M}{A}},\phi^{-1R}_{\actR{M}{B}}) \circ P(\phi_{I,M,A},\phi^{-1R}_{I,M,B}) = P(\lambda^{-1}_M \oslash A, \lambda_M \oslash B).
\end{align*}
Now, by the wedge condition of the end, the left-bottom side
of the previous diagram is just the projection \(\pi_M\).
\[\begin{tikzcd}[column sep=small]
    & \Theta P(A,B) \dlar[swap]{\pi_{I \otimes M}} \drar{\pi_M}& \\
    P[I \otimes M](A,B) \drar[swap]{P(\id,\lambda_M \oslash B)} && P[M](A,B) \dlar{P(\lambda_M \oslash A,\id)} \\
    & P(\actL{I \otimes M}{A},\actR{M}{B}) &
  \end{tikzcd}\]
Finally, by the universal property of the end, that implies
that \(\Theta \varepsilon_P \circ \delta_P\) must coincide with the identity.

Let us now prove right counitality,
\(\varepsilon_{\Theta P} \circ \delta_P = \id_{\Theta P}\), which follows from
the commutativity of the following diagram.
\[\begin{tikzcd}
    \Theta P(A,B) \rar{\delta_P} \ar{dd}[swap]{\pi_{I \otimes M}} &
    \Theta\Theta P(A,B) \rar{ \varepsilon_{\Theta P}} \dar{\pi_I} &
    \Theta P(A,B) \dar{\id} \\ &
    \Theta P[I](A,B) \dar{\pi_M} \rar{\Theta P(\phi_I, \phi_I^{-1R})} &
    \Theta P(A,B) \dar{\pi_M} \\
    P[M \otimes I](A,B) \ar[swap]{r}[yshift=-1ex]{P(\phi_{M,I,A},\phi^{-1R}_{M,I,B})} &
    P[I][M](A,B) \ar[swap]{r}[yshift=-1ex]{P(\phi_A,\phi_B^{-1R})} &
    P[M](A,B)
  \end{tikzcd}\]
Again, the definition of \(\delta\) makes the left pentagon
commute. The upper right square commutes now because of the definition of
\(\varepsilon\), whereas the lower right square commutes because of
functoriality of ends and naturality of \(\pi\). By the axioms of the monoidal
actions (Remark~\ref{remark:axiomsmonoidal}), the bottom edge of the diagram
can be rewritten as
\[
  P(\phi_A^L,\phi_B^{-1R}) \circ P(\phi_{M,I,A}^L,\phi^{-1R}_{M,I,B}) = P(\rho_M^{-1} \oslash A, \rho_M \oslash B).
\]
Now, by the wedge condition of the end, the left-bottom side
of the previous diagram is just the projection \(\pi_M\).
\[\begin{tikzcd}[column sep=small]
 & \Theta P(A,B) \dlar[swap]{\pi_{I \otimes M}} \drar{\pi_M}& \\
 P[I \otimes M](A,B) \drar[swap]{P(\id,\lambda_M \oslash B)} && P[M](A,B) \dlar{P(\lambda_M \oslash A,\id)} \\
 & P(\actL{I \otimes M}{A},\actR{M}{B}) &
\end{tikzcd}\]
Finally, by the universal property of the end,
\(\varepsilon_{\Theta P} \circ \delta_P\) must coincide with the identity.

Coassociativity,
\(\Theta\delta_P\circ\delta_{P} = \delta_{\Theta P}\circ\delta_P\), follows from
commutativity of the following diagram in \Cref{fig:diagramassoc}.

\begin{figure}[!h]
  \centering
\begin{tikzcd}
&    \Theta P(A,B)
      \ar{rd}{\delta_{P}}
      \ar[swap]{dl}{\delta_{P}}
& \\ \Theta\Theta P(A,B)
        \ar[swap]{rd}{\Theta\delta_{P}}
        \ar[swap]{ddd}{\pi_K}
 &&   \Theta\Theta P(A,B)
        \ar{ddd}{\pi_{N \otimes K}}
        \ar{ld}{\delta_{\Theta P}}
\\ &  \Theta \Theta \Theta P(A,B)
       \ar{d}{\pi_K}
\\ &  \Theta \Theta P[K](A,B)
       \ar{dd}{\pi_N}
& \\  \Theta P[K](A,B)
       \urar{\delta_{P[K]}}
       \ar[swap]{ddd}{\pi_{M \otimes N}}
&&   \Theta P[N \otimes K](A,B)
       \ar{ddd}{\pi_M}
       \dlar[near end]{\Theta P(\phi_{N,K,A},\phi_{N,K,B}^{-1R})}
\\ &  \Theta P[K][N](A,B)
       \dar{\pi_M}
&\\ & P[K][N][M](A,B)
\\   P[K][M \otimes N](A,B)
       \urar[near end,xshift=0.5em]{P(\phi_{M,N,\actL{K}{A}}, \phi_{M,N,\actR{K}{B}}^{-1R})}
       \dar{P(\phi_{M \otimes N,K,A},\phi^{-1R}_{M \otimes N,K,B})}
&&   P[N \otimes K][M](A,B)
       \ular[swap, near end]{P(\phi_{N,K,A},\phi_{N,K,B}^{-1R})}
       \dar[swap]{P(\phi_{M, N\otimes K,A},\phi^{-1R}_{M, N\otimes K,B})}
\\ P[(M\otimes N)\otimes K](A,B) \ar{rr}{P(a^{-1}_{M,N,K},a_{M,N,K})} && P[M\otimes (N\otimes K)](A,B)
\end{tikzcd}
\caption{Diagram for the coassociativity axiom.}\label{fig:diagramassoc}
\end{figure}

We need to show that the upper diamond commutes; by the universal property of
the ends, this amounts to showing that it commutes when followed by
\(\pi_M \circ \pi_N \circ \pi_K\). The lower pentagon is made of isomorphisms,
and it commutes by the axioms of the monoidal actions
(\Cref{remark:axiomsmonoidal}). The two upper degenerate pentagons commute by
definition of \(\delta\). The two trapezoids commute by functoriality of ends
and naturality of the projections.

Finally, the two outermost morphisms of the diagram correspond to two
projections from \(\Theta P(A,B)\), namely \(\pi_{(M \otimes N) \otimes K}\) and
\(\pi_{M \otimes (N \otimes K)}\). The wedge condition for the associator
\(a_{M,N,K}\) makes the external diagram commute.
\end{proof}

\begin{lemma}\label{lemma:tambaracoalg}
  Tambara modules are precisely coalgebras for this comonad. There exists an
  isomorphism of \Vt{categories} \(\Tamb \cong EM(\Theta)\) from the
  category of Tambara modules to the Eilenberg-Moore category of \(\Theta\).
\end{lemma}
\begin{proof}
  Note that the object of \Vt{natural} transformations from \(P\) to
  \(\Theta P\) is precisely
  \begin{align*}
  & \Prof(\C,\D)(P,\Theta P) \\
  \cong & \hintNaturalTransformation \\
  & \int_{A,B} \V\left(P(A,B),\int_{M \in \M} P(\actL{M}{A}, \actR{M}{B})\right) \\
  \cong & \hintContinuity \\
  & \int_{A,B,M} \V\left(P(A,B), P(\actL{M}{A}, \actR{M}{B})\right)
  \end{align*}
  whose elements can be seen as a family of morphisms that is dinatural in
  \(M \in \M\) and natural in \((A,B) \in \C^{op} \otimes \D\). The two conditions in the
  definition of Tambara module can be rewritten as the axioms of the coalgebra.
\end{proof}

\begin{proposition}
  The \(\Theta\) comonad has a left \Vt{adjoint} \(\Phi\), which must therefore
  be a monad. On objects, it is given by the following formula.
\[\Phi Q(X,Y) = \int^{M,U,V}
  Q(U,V) \otimes \C(X,\actL{M}{U}) \otimes \D(\actR{M}{V},Y).\]
That is, there
exists a \Vt{natural} isomorphism \(\Nat(\Phi Q,P) \cong \Nat(Q,\Theta P)\).
\end{proposition}

\begin{proof}
  We can explicitly construct the \Vt{natural} isomorphism using coend calculus.
\begin{align*}
& \int_{A,B} \V \left(Q(A,B), \int_{M}P(\actL{M}{A}, \actR{M}{B})\right)
\\ \cong & \hintContinuity
\\ & \int_{M,A,B} \V(Q(A,B), P(\actL{M}{A}, \actR{M}{B}))
\\ \cong & \hintYoneda
\\ & \int_{M,A,B} \V\left( Q(A,B) , \int_{X,Y} \V \Big(\C(X, \actL{M}{A}) \otimes \D(\actR{M}{B},Y) , P(X,Y)\Big) \right)
\\ \cong & \hintContinuity
\\ & \int_{M,A,B,X,Y} \V \left( Q(A,B) , \V \Big(\C(X, \actL{M}{A}) \otimes \D(\actR{M}{B},Y) , P(X,Y)\Big) \right)
\\ \cong & \hint{Copower}
\\ & \int_{M,A,B,X,Y} \V(Q(A,B) \otimes \C(X, \actL{M}{A}) \otimes \D(\actR{M}{B}, Y), P(X,Y))
\\ \cong & \hintContinuity
\\ & \int_{X,Y} \V \left( \int^{M,A,B} Q(A,B) \otimes \C(X, \actL{M}{A}) \otimes \D(\actR{M}{B},Y) , P(X,Y) \right). \\
\end{align*}
Alternatively, the adjunction can be deduced from the definition of the comonad
\(\Theta\) and the characterization of global Kan extensions as adjoints to
precomposition.
\end{proof}

\subsection{Pastro-Street's ``double'' promonad}
The second part of this proof occurs in the bicategory of \Vt{profunctors}. In
this bicategory, there exists a formal analogue of the Kleisli construction that,
when applied to the Pastro-Street monad \(\Phi\), yields a category whose
morphisms are the optics from Definition~\ref{def:optic}. This is the crucial
step of the proof, as the universal property of that Kleisli construction will
imply that copresheaves over the category of optics there defined are Tambara
modules (Lemma~\ref{lemma:copresheaves}). After that, the last step will be a
relatively straightforward application of the Yoneda lemma
(Lemma~\ref{lemma:doubleyoneda}).

Let \(\Prof\) be the bicategory of \Vt{profunctors} that has as 0-cells the
\Vt{categories} \(\C,\D,\E, \dots\); as 1-cells \(P \colon \C \nto \D\) the
\Vt{profunctors} given as \(P \colon \C^{op} \otimes \D \to \V\); and as 2-cells
the natural transformations between them. The composition of two
\Vt{profunctors} \(P \colon \C^{op} \otimes \D \to \V\) and
\(Q \colon \D^{op} \otimes \E \to \V\) is the \Vt{profunctor}
\(Q \diamond P \colon \C^{op} \otimes \E \to \V\) defined on objects by the
coend\footnote{Although in general the composition of two profunctors can fail
  to exist for size reasons or when $\V$ lacks certain colimits, we only ever
  need these composites in a formal sense. This perspective can be formalized
  with the notion of \emph{virtual equipment}~\cite{cruttwell09}.}
\[
  (Q \diamond P)(C,E) = \int^{D \in \D} P(C,D) \otimes Q(D,E).
\]
There is, however, an equivalent way of defining profunctor composition if we
interpret each \Vt{profunctor} \(\C^{op} \otimes \D \to \V\) as a \Vt{functor}
\(\C^{op} \to [\D,\V]\) to the category of copresheaves. In this case, the
composition of two profunctors \(P \colon \C^{op} \to [\D,\V]\) and
\(Q \colon \D^{op} \to [\E,\V]\) is the \Vt{functor}
\((Q \diamond P) \colon \C^{op} \to [\E,\V]\) defined by taking a left Kan
extension \((Q \diamond P) \coloneqq \Lan_y Q \circ P\) along the Yoneda
embedding \(y \colon \D^{op} \to [\D,\V]\). The unit profunctor for composition
is precisely the Yoneda embedding.
\[\begin{tikzcd}
    & \D^{op} \rar{Q} \dar[swap]{y} & \left[\E, \V \right] \\
    \C^{op} \rar{P} & \left[\D, \V \right] \urar[swap]{\Lan_y Q \circ P} &
  \end{tikzcd}\]
In the same way that we can construct a Kleisli category over a
monad, we will perform a Kleisli construction over the monoids of the bicategory
\(\Prof\), which are called \textbf{promonads}.
Promonads over the base category
\(\V\) that are also Tambara modules for the product appear frequently in the
literature on functional programming languages under the name of
\emph{arrows}~\cite{hughes00,jacobs09,rivas14}.

\begin{definition}
  A \textbf{promonad} is given by a \Vt{category} \(\A\), an endoprofunctor
  \(T \colon \A^{op} \otimes \A \to \V\), and two \Vt{natural} families
  \(\eta_{X,Y} \colon \C(X,Y) \to T(X,Y)\) and
  \(\mu_{X,Y} \colon (T \diamond T)(X,Y) \to T(X,Y)\) satisfying the following
  unitality and associativity axioms.
\[\begin{tikzcd}
    T \rar{\eta \diamond \id}\drar[swap]{\id} & T\diamond T \dar{\mu} & \lar[swap]{\id \diamond \eta}\dlar{\id} T &[-2em]
    T \diamond T \diamond T\dar[swap]{\id \diamond \mu} \rar{\mu \diamond \id} & T \diamond T \dar{\mu} \\
  & T && T \diamond T \rar{\mu}& T
\end{tikzcd}\]
A module for the promonad is a \Vt{profunctor}
\(P \colon \X^{op} \otimes \A \to \V\), together with a \Vt{natural}
transformation \(\rho \colon T \diamond P \to P\) making the following diagrams
commute.
\[\begin{tikzcd}
P \rar{\eta \diamond \id} \drar[swap]{\id} & T \diamond P \dar{\rho} &
T \diamond T \diamond P \rar{\mu \diamond \id} \dar[swap]{\id \diamond \rho} & T \diamond P \dar{\rho} \\
& P & T \diamond P \rar{\rho} & P
\end{tikzcd}\]
An algebra is a module structure on a \Vt{copresheaf} \(P \colon \A \to \V\),
interpreted as a profunctor \(\mathbf{I}^{op} \otimes \A \to \V\) from the unit
\Vt{category}.
\end{definition}

\begin{lemma}\label{lemma:kleisli}
  The bicategory \(\Prof\) admits the Kleisli construction~\cite[\S
  6]{pastro08}. The Kleisli \Vt{category} \(\Kl(T)\) for a promonad
  \((T, \mu, \eta)\) over \(\A\) is constructed as having the same objects as
  \(\A\) and letting the hom-object between \(X,Y \in \A\) be precisely
  \(T(X,Y) \in \V\).
\end{lemma}
\begin{proof}
  The multiplication of the promonad is a \Vt{natural} transformation whose
  components can be taken as the definition for the composition of the
  \Vt{category} \(\Kl(T)\). The following isomorphism is a consequence of continuity.
  \begin{align*}
    & \V\left(\int^{Z \in \C} T(X,Z) \otimes T(Z,Y) , T(X,Y)\right) 
      \quad\cong\quad 
      \int_{Z \in \C} \V\left(T(X,Z) \otimes T(Z,Y) , T(X,Y)\right)
  \end{align*}

  Let us show now that this \Vt{category} satisfies the universal property of
  the Kleisli construction. Let \(P \colon \X^{op} \otimes \A \to \V\) be a
  \Vt{profunctor}. A module structure \(\rho \colon T \diamond P \to P\)
  corresponds to a way of making the profunctor \(P\) functorial over \(\Kl(T)\) in the second argument
  \begin{align*}
  & \int_{X \in \X, Z \in \A} \V\left(\int^{Y \in \A} P(X,Y) \otimes T(Y,Z), P(X,Z) \right)\\
  \cong & \hintContinuity \\
  & \int_{X,Y,Z} \V(P(X,Y) \otimes T(Y,Z), P(X,Z)) \\
  \cong & \hint{Exponential} \\
  & \int_{X,Y,Z} \V(T(Y,Z), [P(X,Y), P(X,Z)]).
\end{align*}
Functoriality of this family follows from the monad-algebra axioms.
\end{proof}

\begin{lemma}\label{lemma:copresheaves}
  The category of algebras for a promonad is equivalent to the copresheaf
  category over its Kleisli object.
\end{lemma}
\begin{proof}
  Let \(\X\) be any category and \(\Phi \colon \Y \to \Y\) a promonad. By the
  universal property of the Kleisli construction (see
  Lemma~\ref{lemma:kleisli}), \(\Prof(\X,\operatorname{Kl}(\Phi))\) is
  equivalent to the category of modules on \(\X\) for the promonad. In
  particular, \Vt{profunctors} from the unit \Vt{category} to the Kleisli object
  form precisely the category \(\EM(\Phi)\) of algebras for the promonad; thus
  \[[\operatorname{Kl}(\Phi),\V] \cong [\mathbf{I}^{op} \otimes \operatorname{Kl}(\Phi),\V] \cong \Prof(\mathbf{I},\operatorname{Kl}(\Phi)) \cong \EM(\Phi).\qedhere\]
\end{proof}

\begin{proposition}\label{prop:cocontinous}
  Let \(T \colon [\A,\V] \to [\A,\V]\) be a cocontinuous monad. The profunctor
  \(\check{T} \colon \A^{op} \to [\A,\V]\) defined by
  \(\check{T} \coloneqq T \circ y\) can be given a promonad structure. Moreover,
  algebras for \(T\) are precisely algebras for the promonad \(\check{T}\).
\end{proposition}
\begin{proof}
  First, the fact that \(T\) is cocontinuous means that it preserves left Kan
  extensions and thus,
  \[\Lan_y\check{T} \cong \mathsf{Lan}_y(T \circ y) \cong T \circ \mathsf{Lan}_y(y) \cong T.\]
  This means that the composition of the profunctor \(\check{T}\) with itself is
  \[\check{T} \diamond \check{T} = \mathsf{Lan}_y \check{T} \circ \check{T}
  \cong \mathsf{Lan}_y \check{T} \circ T\circ y
  \cong T \circ T \circ y.
  \]
  The unit and multiplication of the promonad are then obtained by whiskering
  the unit and multiplication of the monad with the Yoneda embedding; that is,
  \((\eta \circ y) \colon y \to T\circ y\) and
  \((\mu \circ y) \colon T \circ T \circ y \to T \circ y\). The diagrams for
  associativity and unitality for the promonad are the whiskering by the Yoneda
  embedding of the same diagrams for the monad. In fact, the same reasoning
  yields that, for any \(P \colon \D^{op} \to [\A, \V]\),
  \[\check{T} \diamond P \cong (T \circ y) \diamond P \cong \Lan_y(T \circ y)\circ P \cong T \circ P.\]
  As a consequence of this for the case
  \(P \colon \mathbf{I} \to [\A^{op},\V]\), any \(T\mbox{-algebra}\) can be seen
  as a \(\check{T}\mbox{-algebra}\) and vice versa. The axioms for the promonad
  structure on \(\check{T}\) coincide with the axioms for the corresponding
  monad on \(T\).
\end{proof}

In particular, the Pastro-Street monad \(\Phi\) is a left adjoint. That implies
that it is cocontinuous and, because of Proposition~\ref{prop:cocontinous}, it
induces a promonad \({\check \Phi} = \Phi \circ y\), having Tambara modules as
algebras. We can compute by the Yoneda lemma that
\[{\check \Phi}(A,B,S,T) = \int^{M} \C(S,\actL{M}{A}) \otimes \D(\actR{M}{B}, T).
\]
This coincides with \Cref{def:optic}. We now define
\(\Optic\) to be the Kleisli \Vt{category} for the promonad \({\check\Phi}\).
 
\subsection{Profunctor representation theorem}
Let us zoom out to the big picture again. It has been observed that optics can
be composed using their profunctor representation; that is, profunctor optics
can be endowed with a natural categorical structure. On the other hand, we have
generalized the double construction in \cite{pastro08} to
abstractly obtain the category \(\Optic\). The final missing piece that makes
both coincide is the \textit{profunctor representation theorem}, which will
justify the profunctor representation of optics and their composition in
profunctor form being the usual function composition.

The profunctor representation theorem for the case \(\V = \Sets\) and non-mixed
optics has been discussed in \cite[Theorem 4.2]{boisseau18}. Although our
statement is more general and the proof technique is different, the main idea is
the same. In both cases, the key insight is the following lemma, already
described by \cite{milewski17}.

\begin{lemma}[``Double Yoneda'' from \cite{milewski17}]
\label{lemma:doubleyoneda}
  For any \Vt{category} \(\A\), the hom-object between \(X\) and \(Y\) is
  \Vt{naturally} isomorphic to the object of \Vt{natural} transformations
  between the functors that evaluate copresheaves in \(X\) and \(Y\); that is,
  \[\A(X,Y) \cong [ [ \A , \V ] , \V ](-(X),-(Y)).\]
  The isomorphism is given by the canonical maps \(\A(X,Y) \to \V(FX,FY)\) for each \(F \in [\A, \V]\).  Its inverse is given by computing its value on the identity on the \(\A(X,-)\) component.
\end{lemma}

\begin{proof}
  In the functor \Vt{category} \([ \A , \V]\), we can apply the Yoneda embedding
  to two representable functors \(\A(Y,-)\) and \(\A(X,-)\) to get
  \[[\A, \V](\A(Y,-) , \A(X,-)) \cong \int_F \V\Big([\A(X,-),F], [\A(Y,-),F] \Big).\]
  Here reducing by Yoneda lemma on both the left hand side and the two arguments
  of the right hand side, we get the desired result.
\end{proof}

\begin{remark} %
As a very simple special case of the Double Yoneda construction, the Haskell type
\begin{lstlisting}
    forall f . Functor f => f a -> f b
\end{lstlisting}
is isomorphic to the simple function type \texttt{a -> b} \cite{milewski17}.
It is straightforward for a functional programmer to construct the two witnesses to the isomorphism: the functorial action in one direction, and instantiation to the identity functor in the other. 
\end{remark}

\begin{theorem}[Profunctor representation theorem]
\label{th:profrep}
\[\int_{P \in \Tamb} \V(P(A,B) , P(S,T)) \cong
\Optic((A,B),(S,T)).
\]
\end{theorem}
\begin{proof}
  We apply Double Yoneda (Lemma~\ref{lemma:doubleyoneda}) to the \Vt{category}
  \(\Optic\) and then use that copresheaves over it are precisely Tambara
  modules (Proposition~\ref{lemma:copresheaves}).
  
  The immediate practical application of this theorem is to justify the following profunctor representation
  \begin{lstlisting}
  forall p . Tambara p => p a b -> p s t
  \end{lstlisting}
for all the optics we've been discussing in this paper, where the \(\mathit{Tambara}\) constraint
is replaced by the class of profunctors preserving the appropriate monoidal action.
For instance, the standard lens
\begin{lstlisting}
  type Lens a b s t = forall p . Cartesian p => p a b -> p s t
\end{lstlisting}
is defined by the class of profunctors preserving
the action defined by the cartesian product, which are Tambara modules for the cartesian product.
\begin{lstlisting}
  class Profunctor p => Cartesian p where
    first'  :: p a b -> p (a, c) (b, c)
    second' :: p a b -> p (c, a) (c, b)
\end{lstlisting}
A similar encoding works for the rest of the optics.

\end{proof}

\section{Conclusions}
\label{sec:conclusions}
We have extended a result by Pastro and Street to a setting that is useful for
optics in functional programming. Using it, we have refined some of the optics
already present in the literature to mixed optics, providing derivations for
each one of them. We have also described new optics.

Regarding functional programming, the work suggests an architecture for a
library of optics that would benefit from these results. Instead of implementing
each optic separately, the general definition can be instantiated in all the 
particular cases. We can then just consider specific functions for constructing
the more common families of optics. Tambara modules can be used to implement
each one of the combinators of the library, ensuring that they work for as many
optics as possible. The interested reader can find the implementation in
Appendix~\ref{sec:implementation}.

Many of the other applications of optics may benefit from the flexibility of
enriched and mixed optics. They may be used to capture some \emph{lens}-like
constructions and provide a general theory of how they should be studied; the
specifics remain as future work.

\subsection{Van Laarhoven encoding}
This paper has focused on the profunctor representation of optics. A similar
representation that also provides the benefit of straightforward composition is
the \textbf{van Laarhoven encoding}~\cite{laarhoven09}. It is arguably less
flexible than the profunctor representation, being based on representable
profunctors, but it is more common in practice. For instance, 

\begin{proposition}[Van Laarhoven-style traversals]
Traversals admit a
different encoding in terms of profunctors represented by an applicative
functor.
\[ \mathbf{Traversal}((A,B),(S,T)) \cong  \int_{F \in \App} \V(A, FB) \otimes \V(S,FT). \]
\end{proposition}
\begin{proof}
\begin{align*}
  & \int_{F \in \App} \V(A, FB) \otimes \V(S,FT) \\
    \cong & \hintYoneda \\
  & \int_{F \in \App} [\V,\V](A \otimes [B,-], F) \otimes \V(S,FT) \\
    \cong & \hint{Adjunction, free applicatives} \\
  & \int_{F \in \App} \App \left(\sum_{n \in \mathbb{N}} A^n \otimes [B^n,-], F \right) \otimes \V(S,FT) \\
    \cong & \hintCoyoneda \\
  & \V\left(S, \sum_{n \in \mathbb{N}} A^n \otimes [B^n , T]\right). \qedhere
\end{align*}
\end{proof}

Exactly the same technique yields lenses and \emph{grates}~\cite{deikun15},
using arbitrary representable or corepresentable profunctors, respectively.

\begin{proposition}[Van Laarhoven lenses \cite{laarhoven09}]
Lenses admit an encoding in terms of functors, or representable profunctors.
\[    \LinearLens_{\otimes,\otimes} ((A, B), (S, T)) \cong
 \int_{F \in [\V,\V]} \V(A, FB) \otimes \V(S,FT).
\]
\end{proposition}
\begin{proof}
  By coend calculus, using the Yoneda lemma.
\begin{align*}
  & \int_{F \in [\V,\V]} \V(A, FB) \otimes \V(S,FT) \\
    \cong & \hintYoneda \\
  & \int_{F \in  [\V,\V]} [\V,\V](A \otimes [B,-], F) \otimes \V(S,FT) \\
    \cong & \hintCoyoneda \\
  & \V\left(S, A \otimes [B , T]\right).  \qedhere
\end{align*}
\end{proof}

\begin{proposition}[Van Laarhoven-style grates]
  Grates admit an encoding in terms of corepresentable functors.
\[    \Grate((A, B), (S, T)) \cong
 \int_{F \in [\V,\V]} \V(FA, B) \otimes \V(FS,T).
\]
\end{proposition}
\begin{proof}
  Again by coend calculus, and using the Yoneda lemma.
\begin{align*}
  & \int_{F \in [\V,\V]} \V(FA, B) \otimes \V(FS,T) \\
    \cong & \hintYoneda \\
  & \int_{F \in  [\V,\V]} [\V,\V](F,[[\bullet,A],B]) \otimes \V(FS,T) \\
    \cong & \hintCoyoneda \\
  & \V\left([[S, A],B] , T\right).  \qedhere
\end{align*}
\end{proof}

\subsection{Related work}

Pastro and Street \cite{pastro08} first described the construction of
\emph{doubles} in their study of Tambara theory. Their results can be reused for
\emph{optics} thanks to the observations of \cite{milewski17}. The profunctor
representation theorem and its implications for functional programming have been
studied by \cite{boisseau18}. We combine their approach with Pastro and Street's
to get a proof of a more general version of this theorem.

The case of mixed optics was first mentioned by Riley \cite[\S 6.1]{riley18}, but his
work targeted a more restricted case. Specifically, the definitions of
\emph{optic} given in the literature \cite{milewski17,boisseau18,riley18}
deal only with the particular case in which \(\V = \Sets\), the categories
\(\C\) and \(\D\) coincide, and the two actions are the same.
Riley derives a class of optics and their
laws~\cite[\S4.4]{riley18} that is closely related to ours in \Cref{sec:opticsforcofree};
our proposal makes stronger assumptions but may be more straightforward to apply in
programming contexts. Riley uses the results of \cite{jaskelioff15} to
propose a description of the traversal in terms of traversable functors~\cite[\S
4.6]{riley18}; our derivation simplifies this approach, which was in principle
not suitable for the enriched case.

A central aspect of Riley's work is the extension of the concept of \emph{lawful
 lens} to arbitrary \emph{lawful optics}~\cite[\S 3]{riley18}. This extension
works exactly the same for the optics we define here, so we do not address it
explicitly in this paper. A first reasonable notion of lawfulness for the case of
mixed optics for two actions \((\actL{}{}) \colon \M \otimes \C \to \C\) and
\((\actR{}{}) \colon \N \otimes \D \to \D\) is to use a cospan
\(\C \to \mathbf{E} \gets \D\) \emph{of actions} to push the two parts of the
optic into the same category and then consider lawfulness in \(\mathbf{E}\).

\subsection{Further work}

In terms of functional programming, optics of different kinds compose using polymorphic function composition.
 A categorical account of how optics of different kinds compose into
    optics is left for further work. Specifically, it should be able to explain
    the ``lattice of optics'' described in \cite{pickering17,boisseau18}.
    Some preliminary results have been discussed by \cite{roman19}, but
    the proposal to model the lattice is still too \emph{ad-hoc} to be
    satisfactory.  The topic of lawfulness \cite[\S 3]{riley18} and how it
    relates to composition and mixed optics is also left for further work.

 The relation between power series functors and traversables is implicit
    across the literature on polynomial functors and containers. It can be shown
    that \emph{traversable} structures over an endofunctor \(T\) correspond to
    certain parameterised coalgebras using the free applicative
    construction~\cite{rypacek12}. We believe that it is possible to refine this
    result to make our derivation for traversals more practical for functional
    programming.

    It can be noted that lenses are the optic for products, functors that
    distribute over strong functors. Traversals are the optic for traversables,
    functors that distribute over applicative functors. Both have a van
    Laarhoven representation in terms of strong and applicative functors
    respectively. A generalization of this phenomenon needs a certain Kan
    extension to be given a coalgebra structure~\cite[Lemma 4.1.3]{roman19}, but
    it does not necessarily work for any optic.

 Optics have numerous applications in the literature, including game
    theory~\cite{ghani18}, machine learning~\cite{fong19} and model-driven
    development~\cite{stevens10}. Beyond functional programming, enriched optics
    open new paths for explorating applications of optics. Both mixed optics
    and enriched optics allow us to more precisely adjust the existing
    definitions to match the desired applications.

\newpage

\section{Appendix: Haskell implementation}
\label{sec:implementation}

Let \(\V\) be a cartesian closed category whose objects model the types of our
programming language and whose points \(1 \to X\) represent programs of type
\(X\). The following is an informal translation of the concepts of enriched
category theory to a Haskell implementation where a single abstract definition
of optic is used for a range of different examples. The code for this text can
be compiled under GHC 8.6, using the libraries \hask{split} and \hask{delay}. It
includes an implementation of optics and all the examples we have discussed
(Figures~\ref{fig:example0},~\ref{fig:exampleBox},~\ref{fig:exampleIris} and
\ref{fig:exampleTraversal}).

The complete code can be found at
\begin{center}
  \url{https://github.com/mroman42/vitrea}
\end{center}

\subsection{Concepts of enriched category theory}

\begin{definition}[{{\cite[\S 1.2]{kelly05}}}, see also {{\cite{visscher15}}}]
  A \Vt{category} \(\C\) consists of a set \(\mathrm{Obj}(\C)\) of objects, a
  hom-object \(\C(A,B)\in \V\) for each pair of objects
  \(A,B \in \mathrm{Obj}(\C)\), a composition law
  \(\C(A,B) \times \C(B,C) \to \C(A,C)\) for each triple of objects, and an
  identity element \(1 \to \C(A,A)\) for each object; subject to the usual
  associativity and unit axioms.
\end{definition}
\begin{lstlisting}
class Category objc c where
  unit :: (objc x) => c x x
  comp :: (objc x) => c y z -> c x y -> c x z
\end{lstlisting}
In Haskell, we define a class of higher-order type constructors that represent a category.
Here, the objects for our category are selected from Haskell types by the constraint \hask{objc}. 
Hom-objects are selected by the two-argument type constructor \hask{c}. We then ask that they have a identity, \hask{unit}, and a composition, \hask{comp}.

\begin{definition}[{{\cite[\S 1.2]{kelly05}}}]
  A \Vt{functor} \(F \colon \C \to \D\) consists of a function
  \(\Obj(\C) \to \Obj(\D)\) together with a map \(\C(A,B) \to \D(FA,FB)\) for
  each pair of objects; subject to the usual compatibility with composition and
  units. \Vt{bifunctors} and \Vt{profunctors} can be defined analogously. 
\end{definition}
\begin{lstlisting}
class ( Category objc c, Category objd d, Category obje e
      , forall x y . (objc x , objd y) => obje (f x y) )
      => Bifunctor objc c objd d obje e f where
  bimap :: ( objc x1, objc x2, objd y1, objd y2 )
        => c x1 x2 -> d y1 y2 -> e (f x1 y1) (f x2 y2)

class ( Category objc c, Category objd d )
      => Profunctor objc c objd d p where
  dimap :: (objc x1, objc x2, objd y1, objd y2)
        => c x2 x1 -> d y1 y2 -> p x1 y1 -> p x2 y2
\end{lstlisting}
Bifunctors and profunctors are both instances of functors.
Here, we ask for the maps on morphisms, \hask{bimap} for bifunctors and \hask{dimap} for profunctors.

\begin{definition}[{{\cite{day70}}}]
  A monoidal \Vt{category} is a \Vt{category} \(\M\) together with a
  \Vt{functor} \((\otimes) \colon \M \otimes \M \to \M\), an object
  \(I \in \M\), and \Vt{natural} isomorphisms
  \(\alpha \colon (A \otimes B) \otimes C \cong A \otimes (B \otimes C)\),
  \(\rho \colon A \otimes I \cong A\), and
  \(\lambda \colon I \otimes A \cong A\), satisfying the usual coherence axioms
  for a monoidal category.
\end{definition}
\begin{lstlisting}
class ( Category obja a
      , Bifunctor obja a obja a obja a o
      , obja i )
      => MonoidalCategory obja a o i where
  alpha  :: (obja x, obja y, obja z)
         => a (x `o` (y `o` z)) ((x `o` y) `o` z)
  alphainv  :: (obja x, obja y, obja z)
            => a ((x `o` y) `o` z) (x `o` (y `o` z))
  lambda    :: (obja x) => a (x `o` i) x
  lambdainv :: (obja x) => a x (x `o` i)
  rho       :: (obja x) => a (i `o` x) x
  rhoinv    :: (obja x) => a x (i `o` x)
\end{lstlisting}
We define the monoidal category relative to a bifunctor representing the tensor and an object representing the unit.
We ask for all the coherence maps, i.e. \hask{alpha} for the associator.

\begin{definition}
  A monoidal \Vt{action} \(F \colon \M \otimes \C \to \C\) from a
  monoidal \Vt{category} \(\M\) to an arbitrary category \(\C\) is a
  \Vt{functor} together with two \Vt{natural} isomorphisms \(F(I,X) \cong X\)
  and \(F(M,F(N,X))\cong F((M \otimes N), X)\) satisfying associativity and
  unitality conditions.
\end{definition}
\begin{lstlisting}
class ( MonoidalCategory objm m o i
      , Bifunctor objm m objc c objc c f
      , Category objc c )
      => MonoidalAction objm m o i objc c f where
  unitor :: (objc x) => c (f i x) x
  unitorinv :: (objc x) => c x (f i x)
  multiplicator :: (objc x, objm p, objm q)
                => c (f p (f q x)) (f (p `o` q) x)
  multiplicatorinv :: (objc x, objm p, objm q)
                => c (f (p `o` q) x) (f p (f q x))
\end{lstlisting}
Monoidal actions are bifunctors with extra maps representing the coherence conditions, e.g. \hask{unitor} for the oplaxator.

\begin{definition}
  \label{def_haskell_optic}
  \Cref{def:optic} has now a direct interpretation in more generality.
  Note how the coend is modeled as an existential type in \hask{x} using a
  \texttt{GADT}.
\begin{lstlisting}
  data Optic objc c objd d objm m o i f g a b s t where
    Optic :: ( MonoidalAction objm m o i objc c f
             , MonoidalAction objm m o i objd d g
             , objc a, objc s , objd b, objd t , objm x )
          => c s (f x a) -> d (g x b) t
          -> Optic objc c objd d objm m o i f g a b s t
\end{lstlisting}
\end{definition}

\subsection{Mixed profunctor optics}

We can implement Tambara modules (\Cref{def:tambara}) and profunctor
optics using the profunctor representation theorem (\Cref{th:profrep}).
\begin{lstlisting}
class ( MonoidalAction objm m o i objc c f
      , MonoidalAction objm m o i objd d g
      , Profunctor objc c objd d p )
      => Tambara objc c objd d objm m o i f g p where
  tambara :: (objc x, objd y, objm w)
          => p x y -> p (f w x) (g w y)

type ProfOptic objc c objd d objm m o i f g a b s t = forall p .
  ( Tambara objc c objd d objm m o i f g p
  , MonoidalAction objm m o i objc c f
  , MonoidalAction objm m o i objd d g
  , objc a , objd b , objc s , objd t
  ) => p a b -> p s t
\end{lstlisting}

The isomorphism between existential and profunctor optics can be explicitly
constructed from \Cref{lemma:doubleyoneda}.

\ignore{
\begin{lstlisting}
 -- The translation from profunctor to existential optics requires the
 -- following two isntances to be defined.  The first one describes the
 -- representable copresheaves of the category of optics, that is,
 -- Optic(A,B,-,-).
 instance ( MonoidalAction objm m o i objc c f
          , MonoidalAction objm m o i objd d g
          , objc a , objd b )
          => Profunctor objc c objd d (Optic objc c objd d objm m o i f g a b) where
   dimap f g (Optic l r) = Optic (comp @objc @c l f) (comp @objd @d g r)

 -- The second one shows that these representable copresheaves are
 -- Tambara modules for their defining actions.
 instance ( MonoidalAction objm m o i objc c f
          , MonoidalAction objm m o i objd d g
          , objc a , objd b )
          => Tambara objc c objd d objm m o i f g (Optic objc c objd d objm m o i f g a b) where
   tambara (Optic l r) = Optic
     (comp @objc @c (multiplicator @objm @m @o @i @objc @c @f) (bimap @objm @m @objc @c @objc @c (unit @objm @m) l))
     (comp @objd @d (bimap @objm @m @objd @d @objd @d (unit @objm @m) r) (multiplicatorinv @objm @m @o @i @objd @d @g))
\end{lstlisting}
}

\begin{lstlisting}
 ex2prof :: forall objc c objd d objm m o i f g a b s t .
        Optic     objc c objd d objm m o i f g a b s t
     -> ProfOptic objc c objd d objm m o i f g a b s t
 ex2prof (Optic l r) =
   dimap @objc @c @objd @d l r .
   tambara @objc @c @objd @d @objm @m @o @i

 prof2ex :: forall objc c objd d objm m o i f g a b s t .
     ( MonoidalAction objm m o i objc c f
     , MonoidalAction objm m o i objd d g
     , objc a , objc s
     , objd b , objd t )
     => ProfOptic objc c objd d objm m o i f g a b s t
     -> Optic     objc c objd d objm m o i f g a b s t
 prof2ex p = p (Optic
     (unitorinv @objm @m @o @i @objc @c @f)
     (unitor @objm @m @o @i @objd @d @g))
\end{lstlisting}
We used the \hask{TypeApplications} language extension to explicitly pass type 
parameters to polymorphic functions.

\subsection{Combinators}

After constructing optics, an implementation should provide ways of using them.
Many optics libraries, such as Kmett's \emph{lens}~\cite{kmett15}, provide a
vast range of combinators. Each of these combinators works on some group of
optics that share a common feature. For instance, we could consider all the
optics that implement a \texttt{view} function, and create a single combinator
that lets us view the focus inside a family of optics.

This may seem, at first glance, difficult to model. We do not know, a priori,
which of our optics will admit a given combinator. However, the fact that
Tambara modules are copresheaves over optics suggests that we can use them to
model ways of accessing optics; and in fact, we have found them to be very
satisfactory to describe combinators in their full generality.

\begin{remark}
  As an example, for any fixed \(A\) and \(B\), consider the profunctor
  \(P_{A,B}(S,T) \coloneqq (S \to A)\). It can be seen as modelling the
  \texttt{view} combinator that some optics provide.
\begin{lstlisting}
newtype Viewing a b s t = Viewing { getView :: s -> a }
instance Profunctor Any (->) Any (->) (Viewing a b) where
  dimap l _ (Viewing f) = Viewing (f . l)
\end{lstlisting}
If we want to apply this combinator to a particular optic, we need it to be a
Tambara module for the actions describing the optic. For instance, we can show
that it is a Tambara module for the cartesian product, taking \(\C = \D = \M\);
this means it can be used with \emph{lenses} in the cartesian case. In other
words, \emph{lenses} can be used to \texttt{view} the focus.
\begin{lstlisting}
  instance Tambara Any (->) Any (->) Any (->) (,) ()
      (,) (,) (Viewing a b) where
    tambara (Viewing f) = Viewing (f . snd)
\end{lstlisting}
Optic combinators are usually provided as infix functions that play nicely
with the composition operator. Specifically, they have ``fixity and semantics
such that subsequent field accesses can be performed with
\hask{Prelude..} [function composition]''~\cite{kmett15}.
\begin{lstlisting}
   infixl 8 ^.
  (^.) :: s -> (Viewing a b a b -> Viewing a b s t) -> a
  (^.) s l = getView (l (Viewing id)) s
\end{lstlisting}
\end{remark}

\subsubsection{Table of combinators}
\label{sec:tablecombinators}
The names of our combinators try to match, where possible, the names used by Kmett's lens
library~\cite{kmett15}.

\quad\\

\begin{tabular}{rl}
\hline
\vrule width0pt height2.5ex depth 1ex %
 & Combinators. \\
\hline
\vrule width0pt height2.5ex %
\hask{(^.) ::} & \hask{s -> (Viewing a b a b -> Viewing a b s t) -> a} \\
& View a single target. \\
  \hask{(?.) ::} & \hask{s -> (Previewing a b a b  } \\
  & \hask{-> Previewing a b s t) -> Maybe a} \\
& Try to view a single target; this can possibly result in failure. \\
\hask{(.~) ::} & \hask{(Setting a b a b -> Setting a b s t) -> b -> s -> t} \\
& Replace a target with a given value. \\
\hask{(}\texttt{\%}\hask{~) ::} & \hask{(Replacing a b a b -> Replacing a b s t)} \\
& \hask{-> (a -> b) -> (s -> t)} \\
  & Replace a target by applying a function. \\
  \hask{(.?) ::} & \hask{(Monad m) => (Classifying m a b a b} \\
  & \hask{ -> Classifying m a b s t'') -> b -> m s -> t} \\
& Classifies the target into a complete instance. \\
\texttt{(>-) ::} & \hask{(Aggregating a b a b -> Aggregating a b s t)} \\
& \hask{-> ([a] -> b) -> [s] -> t} \\
& Aggregates the whole structure by aggregating the targets. \\
  \hask{(.!) ::} & \hask{(Monad m)=> (Updating m a b a b} \\
  & \hask{-> Updating m a b s t)-> b -> s -> m t} \\
  & Replaces the target, producing a monadic effect.
\end{tabular}

\subsection{Table of optics}

We can consider all of these optics in the case where some cartesian closed
\(\W\) is both the enriching category and the base for the optic. This case is
of particular interest in functional programming.

\begin{longtable}{llc}
\hline
\vrule width0pt height2.5ex depth 1ex %
Name & Description & Ref. \\
\hline
\vrule width0pt height2.5ex %
Adapter & \hask{(s -> a) , (b -> t)} &  {\ref{def:adapter}} \\
Lens & \hask{(s -> a) , (s -> b -> t)} &  {{\ref{def:lens}}} \\
Algebraic lens & \hask{(s -> a) , (m s -> b -> t)} &  {{\ref{def:algebraiclens}}} \\
Prism & \hask{(s -> Either a t) , (b -> t)} &  {{\ref{def:prism}}} \\
Coalgebraic prism & \hask{(s -> Either a (c t)) , (b -> t)} & {{\ref{def:algebraiclens}}} \\
Grate & \hask{((s -> a) -> b) -> t} &  {{\ref{def:grate}}} \\
Glass & \hask{((s -> a) -> b) -> s -> t} &  {{\ref{def:glass}}} \\
Affine Traversal & \hask{s -> Either t (a , b -> t)} &  {{\ref{def:affine}}} \\
Traversal & \hask{s -> (Vec n a, Vec n b -> t)} &  {{\ref{def:traversal}}} \\
Kaleidoscope & \hask{(Vec n a -> b) -> (Vec n s -> t)} &  {{\ref{def:kaleidoscope}}} \\
Setter & \hask{(a -> b) , (s -> t)} &  {{\ref{def:setter}}}\\
Fold & \hask{s -> [a]} &  {{\ref{def:fold}}}\\
\end{longtable}

\newpage
\bibliographystyle{plainnat}
\bibliography{main}

\end{document}